\documentclass[final,3p]{elsarticle}
 \usepackage{graphics}
 \usepackage{graphicx}
 \usepackage{epsfig}
\usepackage{amssymb}
 \usepackage{amsthm}
 \usepackage{lineno}
 \usepackage{amsmath}
   \numberwithin{equation}{section}
\usepackage{mathrsfs}

\newtheorem{thm}{Theorem}[section]
\newtheorem{cor}[thm]{Corollary}
\newtheorem{lem}[thm]{Lemma}

\newtheorem{defn}[thm]{Definition}

\biboptions{sort&compress,square}
\begin{document}
\begin{frontmatter}
 \author{Jian Wang}
\author{Yong Wang\corref{cor2}}
\ead{wangy581@nenu.edu.cn}
\cortext[cor2]{Corresponding author}

\author{ ChunLing Yang\corref{cor}}

\address{School of Mathematics and Statistics, Northeast Normal University,
Changchun, 130024, P.R.China}
\title{Dirac Operators with Torsion and the Noncommutative  \\ Residue for Manifolds with  Boundary}
\begin{abstract}
In this paper, we get the Kastler-Kalau-Walze theorem associated to Dirac operators with torsion on compact manifolds with  boundary.
We give two kinds of operator-theoretic explanations of the gravitational action in the case
of 4-dimensional compact manifolds with flat boundary. Furthermore, we get
the Kastler-Kalau-Walze type theorem for four dimensional complex manifolds associated with nonminimal operators.
\end{abstract}
\end{frontmatter}
\section{Introduction}
\label{1}
The Dirac operators evolved into an important tool of modern mathematics, occurring for example in index theory, gauge theory, geometric
quantization, etc. Recently, Dirac operators have assumed a significant place in Connes' noncommutative geometry in \cite{Co1} as the main
ingredient in the definition of a K-cycle. Thus disguised, Dirac operators re-enter modern physics, since non-commutative geometry
can be used, e.g. to derive the action of the Standard Model of elementary particles, as shown in \cite{Co3,FGV}. The noncommutative
residue found in \cite{Gu} and \cite{Wo} plays a prominent role in noncommutative geometry.

In \cite{Co1}, Connes used the noncommutative residue to derive a conformal 4-dimensional Polyakov action analogy.
Further, Connes made a challenging observation that the noncommutative residue of the square of the inverse of the
Dirac opertor was proportion to the Einstein-Hilbert action in \cite{Co2}, which we call the Kastler-Kalau-Walze theorem.
In \cite{Ka}, kastelr gave a brute-force proof of this theorem. In \cite{KW}, Kalau and Walze proved this theorem in the
normal coordinates system simultaneously. And then, Ackermann gave a note on a new proof of this theorem by means of the heat kernel expansion in \cite{Ac}.
 For $3,4$-dimensional spin manifolds with boundary, Wang proved a
 Kastler-Kalau-Walze type theorem for the Dirac operator and the signature operator in \cite{Wa3}. In \cite{Wa4},  Wang computed the lower
dimensional volume ${\rm Vol}^{(2,2)}$ for $5$-dimensional and $6$-dimensional spin manifolds with boundary and  also got the
Kastler-Kalau-Walze type theorem in this case. We proved the Kastler-Kalau-Walze type theorems for foliations with or without boundary
associated with sub-Dirac operators for foliations in \cite{WW1}.
In \cite{AT}, Ackermann and Tolksdorf proved a generalized version of the well-known Lichnerowicz formula for the square of the
most general Dirac operator with torsion  $D_{T}$ on an even-dimensional spin manifold associated to a metric connection with torsion.
 Recently, Pf$\ddot{a}$ffle and Stephan considered compact Riemannian spin manifolds without boundary equipped with orthogonal connections,
and investigated the induced Dirac operators in \cite{PS}. In \cite{PS1}, Pf$\ddot{a}$ffle and Stephan considered
orthogonal connections with arbitrary torsion on compact Riemannian manifolds, and for the induced
Dirac operators, twisted Dirac operators and Dirac operators of Chamseddine-Connes type they computed the spectral
action.

The purpose of this paper is to generalized the results in \cite{Wa3}, \cite{PS} and get a  Kastler-Kalau-Walze type theorems
associated with Dirac operators with torsion on compact manifolds with  boundary.
We  derive the gravitational action on boundary by the
noncommutative residue  associated with Dirac operators with torsion.
 For lower dimensional compact Riemannian manifolds with  boundary and complex manifolds with boundary, we compute the lower
dimensional  volume $\widetilde{{\rm Wres}}[\pi^+(D^{*}_{T})^{-p_{1}}\circ\pi^+D_{T}^{-p_{2}}]$,
and we  get the Kastler-Kalau-Walze theorem  for  lower dimensional manifolds with  boundary.
On the other hand, a special case of Dirac operators with torsion is the Dolbeault operator for complex manifolds.
In \cite{WW}, we considered  the nonminimal operators as generalization of De-Rham Hodge operators and got
 the Kastler-Kalau-Walze type theorems for  nonminimal operators. In this paper, we consider the complex analogy of
 nonminimal operators and prove a Kastler-Kalau-Walze type theorems  for complex nonminimal operators.

 This paper is organized as follows: In Section 2, we define lower dimensional volumes of spin manifolds with torsion.
 In Section 3, for $4$-dimensional spin manifolds with boundary and  the associated Dirac operators with torsion $D^{*}_{T}, D_{T}$,
 we compute the lower dimensional volume ${\rm Vol}^{(1,1)}_4$ and get a Kastler-Kalau-Walze type theorem in this case. In
Section 4, two kinds of operator theoretic explanations of the gravitational action for boundary in the case of 4-dimensional manifolds with
 boundary will be given.
 In Section 5 and Section 6, we get the Kastler-Kalau-Walze type theorem for $6$-dimensional spin manifolds with boundary
associated to $(D_{T}^{*})^{2}$ , $D_{T}^{2}$ with torsion and $4$-dimensional spin manifolds with boundary
for the operator $P^{+}D_{T}^{*}D_{T}$.
In Section 7, we investigate  4-dimensional complex manifolds without boundary associated with complex nonminimal operators.

\section{Lower-Dimensional Volumes of Spin Manifolds with Torsion}

 In this section we consider an $n$-dimensional oriented Riemannian manifold $(M, g^{M})$ equipped
with some spin structure. The Levi-Civita connection
$\nabla: \Gamma(TM)\rightarrow \Gamma(T^{*}M\otimes TM)$ on $M$ induces a connection
$\nabla^{S}: \Gamma(S)\rightarrow \Gamma(T^{*}M\otimes S).$ By adding a additional torsion term $t\in\Omega^{1}(M,End TM)$ we
obtain a new covariant derivative
 \begin{equation}
\widetilde{\nabla}:=\nabla+t
\end{equation}
on the tangent bundle $TM$. Since $t$ is really a one-form on $M$ with values in the bundle of skew endomorphism $Sk(TM)$ in \cite{GHV},
$\nabla$ is in fact compatible with the Riemannian metric $g$ and therefore also induces a connection $\widetilde{\nabla}^{S}:=\nabla^{S}+T$
on the spinor bundle. Here $T\in\Omega^{1}(M, End S)$ denotes the `lifted' torsion term $t\in\Omega^{1}(M, End TM)$.

 Next, we will briefly discuss the construction of this connection. Again, we write $\tilde{\nabla}_{X}Y=\nabla_{X}Y+A(X,Y)$
 with the Levi-Civita connection $\nabla$.
For any $X \in T_{p}M$ the endomorphism $A(X,\cdot)$ is skew-adjoint and hence it is an element of
$\mathfrak{so}(T_{p}M) $, we can express it as
\begin{equation}
A(X,\cdot)=\sum_{i<j}\alpha_{ij}E_{i}\wedge E_{j}.
\end{equation}
 Here $E_{i}\wedge E_{j}$ is meant as the endomorphism of $T_{p}M$ defined by $E_{i}\wedge E_{j}$.
 For any $X \in T_{p}M$ one determines the coefficients in (2.2) by
\begin{equation}
\alpha_{ij}=\langle A(X,E_{i}), E_{j}\rangle=A_{XE_{i} E_{j}}.
\end{equation}
 Each  $E_{i}\wedge E_{j}$  lifts to $\frac{1}{2}E_{i}\cdot E_{j}$
in $spin(n)$, and the spinor connection induced by $\tilde{\nabla}$ is locally
given by
 \begin{equation}
\tilde{\nabla}_{X}\psi=\nabla_{X}\psi+\frac{1}{2}\sum_{i<j}\alpha_{ij}E_{i}\cdot E_{j}\psi=
\nabla_{X}\psi+\frac{1}{2}\sum_{i<j}A_{XE_{i} E_{j}}E_{i}\cdot E_{j}\psi.
\end{equation}
The connection given by (2.4) is compatible with the metric on spinors and with
Clifford multiplication. Then, the Dirac operator associated to the spinor connection from (2.4) is defined as
 \begin{eqnarray}
D_{T}\psi&=&\sum_{i=1}^{n}E_{i}\tilde{\nabla}_{E_{i}}\psi=D\psi+\frac{1}{2}\sum_{i=1}^{n}\sum_{j<k}A_{E_{i} E_{j}E_{k}}E_{i}\cdot E_{j}
\cdot E_{k}\psi\nonumber\\
&=&D\psi+\frac{1}{4}\sum_{i,j,k=1}^{n}A_{E_{i} E_{j}E_{k}}E_{i}\cdot E_{j}\cdot E_{k}\psi
\end{eqnarray}
where $D$ is the Dirac operator induced by the Levi-Civita connection and $``\cdot"$ is the
Clifford multiplication. As Clifford multiplication by any 3-form is self-adjoint we have
 \begin{equation}
D_{T}\psi=D\psi+\frac{3}{2}T\cdot\psi-\frac{n-1}{2}V\cdot\psi,~~D_{T}^{*}\psi=D\psi+\frac{3}{2}T\cdot\psi+\frac{n-1}{2}V\cdot\psi.
\end{equation}
Using the fact that the Clifford multiplication by the vector field $V$ is skew-adjoint,
the hermitian product on the spinor bundle one observes
that $D_{T}$ is symmetric with respect to the natural $L^{2}$-scalar product on spinors if and only
if the vectorial component of the torsion vanishes, $V\equiv 0$. Note that the Cartan type torsion $S$ does not contribute to the Dirac operator
$D_{T}$. As $D_{T}^{*}D_{T}$ is a generalized Laplacian,  one has  the following Lichnerowicz formula.

\begin{thm}\cite{PS}
For the Dirac operator $D_{T}$  associated to the orthogonal connection $\tilde{\nabla}$, we have
 \begin{eqnarray}
D_{T}^{*}D_{T}\psi&=&\Delta\psi+\frac{1}{4}R^{g}\psi+\frac{3}{2}dT\cdot\psi-\frac{3}{4}\parallel T\parallel^{2}\psi\nonumber\\
          &&+\frac{n-1}{2}div^{g}(V)\psi+(\frac{n-1}{2})^{2}(2-n)|V|^{2}\psi\nonumber\\
          &&+3(n-1)(T\cdot V\cdot\psi+(V_{\rfloor}T)\cdot\psi),
\end{eqnarray}
for any spinor field $\psi$, where $\Delta$ is the Laplacian associated to the connection
 \begin{equation}
\tilde{\nabla}_{X}\psi=\nabla_{X}\psi+\frac{3}{2}(X_{\rfloor}T)\cdot\psi-\frac{n-1}{2}V\cdot X\cdot\psi-\frac{n-1}{2}\langle V, X\rangle\psi.
\end{equation}
\end{thm}

To define lower dimensional volume $Vol_{n}^{p_{1},p_{2}}M:=\widetilde{Wres}[\pi^{+}(D_{T}^{*})^{-p_{1}} \circ\pi^{+}D_{T}^{-p_{2}}]$,
 some basic facts and formulae about Boutet de Monvel's calculus can be find in Sec.2 in \cite{Wa1}.
Let $M$ be an n-dimensional compact oriented manifold with boundary $\partial M$. We assume that the metric $g^{M}$ on $M$ has
the following form near the boundary
 \begin{equation}
 g^{M}=\frac{1}{h(x_{n})}g^{\partial M}+\texttt{d}x _{n}^{2} ,
\end{equation}
where $g^{\partial M}$ is the metric on $\partial M$. Let $U\subset
M$ be a collar neighborhood of $\partial M$ which is diffeomorphic $\partial M\times [0,1)$. By the definition of $C^{\infty}([0,1))$
and $h>0$, there exists $\tilde{h}\in C^{\infty}((-\varepsilon,1))$ such that $\tilde{h}|_{[0,1)}=h$ and $\tilde{h}>0$ for some
sufficiently small $\varepsilon>0$. Then there exists a metric $\hat{g}$ on $\hat{M}=M\bigcup_{\partial M}\partial M\times
(-\varepsilon,0]$ which has the form on $U\bigcup_{\partial M}\partial M\times (-\varepsilon,0 ]$
 \begin{equation}
\hat{g}=\frac{1}{\tilde{h}(x_{n})}g^{\partial M}+\texttt{d}x _{n}^{2} ,
\end{equation}
such that $\hat{g}|_{M}=g$.
We fix a metric $\hat{g}$ on the $\hat{M}$ such that $\hat{g}|_{M}=g$.
Note $D_{T}$ is the most general Dirac operator on the spinor bundle $S$ corresponding to a metric connection $\widetilde{\nabla}$ on $TM$.
Let $p_{1},p_{2}$ be nonnegative integers and $p_{1}+p_{2}\leq n$. From Sec 2.1 of \cite{Wa3}, we have
\begin{defn} Lower-dimensional volumes of spin manifolds with boundary with torsion are defined by
   \begin{equation}\label{}
   Vol_{n}^{\{p_{1},p_{2}\}}M:=\widetilde{Wres}[\pi^{+}(D_{T}^{*})^{-p_{1}} \circ\pi^{+}D_{T}^{-p_{2}}] .
\end{equation}
\end{defn}

  Denote by $\sigma_{l}(A)$ the $l$-order symbol of an operator A. An application of (2.1.4) in \cite{Wa1} shows that
\begin{equation}
\widetilde{Wres}[\pi^{+}(D_{T}^{*})^{p_{1}} \circ\pi^{+}D_{T}^{p_{2}}]=\int_{M}\int_{|\xi|=1}\texttt{trace}_{S(TM)}
  [\sigma_{-n}((D_{T}^{*})^{-p_{1}} \circ D_{T}^{-p_{2}})]\sigma(\xi)\texttt{d}x+\int_{\partial M}\Phi,
\end{equation}
where
 \begin{eqnarray}
\Phi&=&\int_{|\xi'|=1}\int_{-\infty}^{+\infty}\sum_{j,k=0}^{\infty}\sum \frac{(-i)^{|\alpha|+j+k+\ell}}{\alpha!(j+k+1)!}
trace_{S(TM)}[\partial_{x_{n}}^{j}\partial_{\xi'}^{\alpha}\partial_{\xi_{n}}^{k}\sigma_{r}^{+}((D_{T}^{*})^{-p_{1}})(x',0,\xi',\xi_{n})\nonumber\\
&&\times\partial_{x_{n}}^{\alpha}\partial_{\xi_{n}}^{j+1}\partial_{x_{n}}^{k}\sigma_{l}(D_{T}^{-p_{2}})(x',0,\xi',\xi_{n})]
\texttt{d}\xi_{n}\sigma(\xi')\texttt{d}x' ,
\end{eqnarray}
and the sum is taken over $r-k+|\alpha|+\ell-j-1=-n,r\leq-p_{1},\ell\leq-p_{2}$.

 \section{The Kastler-Kalau-Walze theorem for $4$-dimensional spin manifolds with boundary  about Dirac Operators
 with torsion $D_{T}^{*}$, $D_{T}$ }
In this section, we compute the lower dimensional volume for 4-dimension compact manifolds with boundary and get a
Kastler-Kalau-Walze type formula in this case.

From now on we always assume that $M$ carries a spin structure so that the spinor bundle is defined and so are
Dirac operator, twisted or generalised Dirac operators on $M$. Connes¡¯ spectral action principle in \cite{Co2}
states that one can extract any action functional of interest in physics from the spectral data of a Dirac operator.

In the following we consider various Dirac operators $D_{T}$ induced by orthogonal connections with general torsion
as in (2.4). We will consider $D_{T}^{*}D_{T}$(since $D_{T}$ is not selfadjoint in general) and the corresponding Seeley-deWitt
coefficients. The Chamseddine-Connes spectral action of $D_{T}^{*}D_{T}$ is
determined if one knows the second and the fourth Seeley-deWitt coefficient.
Since $[\sigma_{-n}((D_{T}^{*})^{-p_{1}} \circ D_{T}^{-p_{2}})]|_{M}$ has
the same expression as $[\sigma_{-n}((D_{T}^{*})^{-p_{1}} \circ D_{T}^{-p_{2}})]|_{M}$ in the case of manifolds without boundary,
so locally we can use the
computations Proposition 3.1 in \cite{PS1} to compute the first term.

\begin{thm}\cite{PS1}
 Let M be a 4-dimensional compact manifold without boundary and $\tilde{\nabla}$ be an orthogonal
connection with torsion. Then we get the volumes  associated to $D_{T}^{*}D_{T}$ on compact manifolds without boundary
 \begin{equation}
Wres((D_{T}^{*}D_{T})^{-1})
  =-\frac{1}{48\pi^{2}}\int_{M}\tilde{R}(x)dx,
\end{equation}
where $\tilde{R}=R+18div(V)-54|V|^{2}-9\parallel T\parallel^{2}$ and $\int_{M}div(V) dVol_{M}=-\int_{\partial_{M}}g(n,V)dVol_{\partial_{M}}$.
\end{thm}

\begin{thm}\cite{PS1}
 Let M be a 6-dimensional compact manifold and $\tilde{\nabla}$ be an orthogonal
connection with torsion. Then we get the volume associated to $D_{T}^{*}D_{T}$ on compact manifolds without boundary
 \begin{equation}
Wres((D_{T}^{*}D_{T})^{-1})=\frac{11}{720}\mathcal{X}(M)-\frac{1}{360\pi^{2}}\int_{M}\parallel C^{g}\parallel^{2}dx
  -\frac{3}{32\pi^{2}}\int_{M}\Big(\parallel \delta T\parallel^{2}+\parallel d(V)\parallel^{2}\big)dx,
\end{equation}
where $C^{g}$ denote the Weyl curvature of the Levi-Civita connection  and
 $\mathcal{X}(M)$ denote the Euler characteristic of $M$.
\end{thm}

So we only need to compute $\int_{\partial M}\Phi$. Let $n=4$, our computation extends to general n.
Let $$ F:L^2({\bf R}_t)\rightarrow L^2({\bf R}_v);~F(u)(v)=\int e^{-ivt}u(t)dt$$ denote the Fourier transformation and
$\Phi(\overline{{\bf R}^+}) =r^+\Phi({\bf R})$ (similarly define $\Phi(\overline{{\bf R}^-}$)), where $\Phi({\bf R})$
denotes the Schwartz space and
  \begin{equation}
r^{+}:C^\infty ({\bf R})\rightarrow C^\infty (\overline{{\bf R}^+});~ f\rightarrow f|\overline{{\bf R}^+};~
 \overline{{\bf R}^+}=\{x\geq0;x\in {\bf R}\}.
\end{equation}
We define $H^+=F(\Phi(\overline{{\bf R}^+}));~ H^-_0=F(\Phi(\overline{{\bf R}^-}))$ which are orthogonal to each other. We have the following
 property: $h\in H^+~(H^-_0)$ iff $h\in C^\infty({\bf R})$ which has an analytic extension to the lower (upper) complex
half-plane $\{{\rm Im}\xi<0\}~(\{{\rm Im}\xi>0\})$ such that for all nonnegative integer $l$,
 \begin{equation}
\frac{d^{l}h}{d\xi^l}(\xi)\sim\sum^{\infty}_{k=1}\frac{d^l}{d\xi^l}(\frac{c_k}{\xi^k})
\end{equation}
as $|\xi|\rightarrow +\infty,{\rm Im}\xi\leq0~({\rm Im}\xi\geq0)$.

 Let $H'$ be the space of all polynomials and $H^-=H^-_0\bigoplus H';~H=H^+\bigoplus H^-.$ Denote by $\pi^+~(\pi^-)$ respectively the
 projection on $H^+~(H^-)$. For calculations, we take $H=\widetilde H=\{$rational functions having no poles on the real axis$\}$ ($\tilde{H}$
 is a dense set in the topology of $H$). Then on $\tilde{H}$,
 \begin{equation}
\pi^+h(\xi_0)=\frac{1}{2\pi i}\lim_{u\rightarrow 0^{-}}\int_{\Gamma^+}\frac{h(\xi)}{\xi_0+iu-\xi}d\xi,
\end{equation}
where $\Gamma^+$ is a Jordan close curve included ${\rm Im}\xi>0$ surrounding all the singularities of $h$ in the upper half-plane and
$\xi_0\in {\bf R}$. Similarly, define $\pi^{'}$ on $\tilde{H}$,
 \begin{equation}
\pi'h=\frac{1}{2\pi}\int_{\Gamma^+}h(\xi)d\xi.
\end{equation}
So, $\pi'(H^-)=0$. For $h\in H\bigcap L^1(R)$, $\pi'h=\frac{1}{2\pi}\int_{R}h(v)dv$ and for $h\in H^+\bigcap L^1(R)$, $\pi'h=0$.
Denote by $\mathcal{B}$ Boutet de Monvel's algebra (for details, see \cite{Wa1} p.735), now we recall the main theorem in \cite{FGLS} p.29.

\begin{thm}\label{th:32}{\bf(Fedosov-Golse-Leichtnam-Schrohe)}
 Let $X$ and $\partial X$ be connected, ${\rm dim}X=n\geq3$,
 $A=\left(\begin{array}{lcr}\pi^+P+G &   K \\
T &  S    \end{array}\right)$ $\in \mathcal{B}$ , and denote by $p$, $b$ and $s$ the local symbols of $P,G$ and $S$ respectively.
 Define:
 \begin{eqnarray}
{\rm{\widetilde{Wres}}}(A)&=&\int_X\int_{\bf S}{\rm{tr}}_E\left[p_{-n}(x,\xi)\right]\sigma(\xi)dx \nonumber\\
&&+2\pi\int_ {\partial X}\int_{\bf S'}\left\{{\rm tr}_E\left[({\rm{tr}}b_{-n})(x',\xi')\right]+{\rm{tr}}
_F\left[s_{1-n}(x',\xi')\right]\right\}\sigma(\xi')dx',
\end{eqnarray}
Then~~ a) ${\rm \widetilde{Wres}}([A,B])=0 $, for any
$A,B\in\mathcal{B}$;~~ b) It is a unique continuous trace on
$\mathcal{B}/\mathcal{B}^{-\infty}$.
\end{thm}

 Recall the definition of the Dirac operator D in \cite{Y}. Denote by $\sigma_{l}(A)$ the $l$-order symbol of
an operator A. In the local coordinates $\{x_{i}; 1\leq i\leq n\}$ and the fixed orthonormal frame
$\{\widetilde{e_{1}},\cdots, \widetilde{e_{n}}\}$, the connection matrix $(\omega_{s,t})$ is defined by
\begin{equation}
\widetilde{\nabla}(\widetilde{e_{1}},\cdots, \widetilde{e_{n}})=(\widetilde{e_{1}},\cdots, \widetilde{e_{n}})(\omega_{s,t}).
\end{equation}
The Dirac operator
\begin{equation}
D=\sum_{i=1}^{n}c(\widetilde{e_{i}})[\widetilde{e_{i}}-\frac{1}{4}\sum_{s,t}\omega_{s,t}(\widetilde{e_{i}})c(\widetilde{e_{s}})c(\widetilde{e_{t}})],
\end{equation}
where $c(\widetilde{e_{i}})$ denotes the Clifford action. Then
\begin{eqnarray}
D_{T}&=&\sum_{i=1}^{n}c(\widetilde{e_{i}})[\widetilde{e_{i}}-\frac{1}{4}\sum_{s,t}\omega_{s,t}(\widetilde{e_{i}})c(\widetilde{e_{s}})
c(\widetilde{e_{t}})]
  +\frac{1}{4}\sum_{i\neq s\neq t}A_{ist}
  c(\widetilde{e_{i}})c(\widetilde{e_{s}})c(\widetilde{e_{t}})  \nonumber\\
  &&+\frac{1}{4}\sum_{i, s, t}[-A_{iit}c(\widetilde{e_{t}})
  +A_{isi}c(\widetilde{e_{s}}) -A_{iss}c(\widetilde{e_{i}})+2A_{iii}c(\widetilde{e_{i}})], \\
D_{T}^{*}&=&\sum_{i=1}^{n}c(\widetilde{e_{i}})[\widetilde{e_{i}}-\frac{1}{4}\sum_{s,t}\omega_{s,t}(\widetilde{e_{i}})c(\widetilde{e_{s}})
c(\widetilde{e_{t}})]
  +\frac{1}{4}\sum_{i\neq s\neq t}A_{ist}
  c(\widetilde{e_{i}})c(\widetilde{e_{s}})c(\widetilde{e_{t}})  \nonumber\\
  &&-\frac{1}{4}\sum_{i, s, t}[-A_{iit}c(\widetilde{e_{t}})
  +A_{isi}c(\widetilde{e_{s}}) -A_{iss}c(\widetilde{e_{i}})+2A_{iii}c(\widetilde{e_{i}})],
  \end{eqnarray}
and
\begin{eqnarray}
\sigma_{1}(D_{T})&=&\sigma_{1}(D_{T}^{*})=\sqrt{-1}c(\xi);\\
\sigma_{0}(D_{T})&=&-\frac{1}{4}\sum_{i,s,t}\omega_{s,t}(\widetilde{e_{i}})c(\widetilde{e_{s}})c(\widetilde{e_{t}})
+\frac{1}{4}\sum_{i\neq s\neq t}A_{ist}
  c(\widetilde{e_{i}})c(\widetilde{e_{s}})c(\widetilde{e_{t}})  \nonumber\\
  &&+\frac{1}{4}\sum_{i, s, t}[-A_{iit}c(\widetilde{e_{t}})
  +A_{isi}c(\widetilde{e_{s}}) -A_{iss}c(\widetilde{e_{i}})+2A_{iii}c(\widetilde{e_{i}})], \\
\sigma_{0}(D_{T}^{*})&=&-\frac{1}{4}\sum_{i,s,t}\omega_{s,t}(\widetilde{e_{i}})c(\widetilde{e_{s}})c(\widetilde{e_{t}})
 +\frac{1}{4}\sum_{i\neq s\neq t}A_{ist}
  c(\widetilde{e_{i}})c(\widetilde{e_{s}})c(\widetilde{e_{t}})  \nonumber\\
  &&-\frac{1}{4}\sum_{i, s, t}[-A_{iit}c(\widetilde{e_{t}})
  +A_{isi}c(\widetilde{e_{s}}) -A_{iss}c(\widetilde{e_{i}})+2A_{iii}c(\widetilde{e_{i}})].
\end{eqnarray}
Hence by Lemma 2.1 in \cite{Wa3}, we have
 \begin{lem}\label{le:31}
The symbol of the Dirac operator
\begin{eqnarray}
\sigma_{-1}(D_{T}^{-1})&=&\sigma_{-1}((D_{T}^{*})^{-1})=\frac{\sqrt{-1}c(\xi)}{|\xi|^{2}}; \\
\sigma_{-2}(D_{T}^{-1})&=&\frac{c(\xi)\sigma_{0}(D_{T})c(\xi)}{|\xi|^{4}}+\frac{c(\xi)}{|\xi|^{6}}\sum_{j}c(\texttt{d}x_{j})
\Big[\partial_{x_{j}}(c(\xi))|\xi|^{2}-c(\xi)\partial_{x_{j}}(|\xi|^{2})\Big]; \\
\sigma_{-2}((D_{T}^{*})^{-1})&=&\frac{c(\xi)\sigma_{0}(D_{T}^{*})c(\xi)}{|\xi|^{4}}+\frac{c(\xi)}{|\xi|^{6}}\sum_{j}c(\texttt{d}x_{j})
\Big[\partial_{x_{j}}(c(\xi))|\xi|^{2}-c(\xi)\partial_{x_{j}}(|\xi|^{2})\Big].
\end{eqnarray}
\end{lem}

Since $\Phi$ is a global form on $\partial M$, so for any fixed point $x_{0}\in\partial M$, we can choose the normal coordinates
$U$ of $x_{0}$ in $\partial M$(not in $M$) and compute $\Phi(x_{0})$ in the coordinates $\widetilde{U}=U\times [0,1)$ and the metric
$\frac{1}{h(x_{n})}g^{\partial M}+\texttt{d}x _{n}^{2}$. The dual metric of $g^{\partial M}$ on $\widetilde{U}$ is
$\frac{1}{\tilde{h}(x_{n})}g^{\partial M}+\texttt{d}x _{n}^{2}.$ Write
$g_{ij}^{M}=g^{M}(\frac{\partial}{\partial x_{i}},\frac{\partial}{\partial x_{j}})$;
$g^{ij}_{M}=g^{M}(d x_{i},dx_{j})$, then

\begin{equation}
[g_{i,j}^{M}]=
\begin{bmatrix}\frac{1}{h( x_{n})}[g_{i,j}^{\partial M}]&0\\0&1\end{bmatrix};\quad
[g^{i,j}_{M}]=\begin{bmatrix} h( x_{n})[g^{i,j}_{\partial M}]&0\\0&1\end{bmatrix},
\end{equation}
and
\begin{equation}
\partial_{x_{s}} g_{ij}^{\partial M}(x_{0})=0,\quad 1\leq i,j\leq n-1;\quad g_{i,j}^{M}(x_{0})=\delta_{ij}.
\end{equation}

Let $\{e_{1},\cdots, e_{n-1}\}$ be an orthonormal frame field in $U$ about $g^{\partial M}$ which is parallel along geodesics and
$e_{i}=\frac{\partial}{\partial x_{i}}(x_{0})$, then $\{\widetilde{e_{1}}=\sqrt{h(x_{n})}e_{1}, \cdots,
\widetilde{e_{n-1}}=\sqrt{h(x_{n})}e_{n-1},\widetilde{e_{n}}=dx_{n}\}$ is the orthonormal frame field in $\widetilde{U}$ about $g^{M}.$
Locally $S(TM)|\widetilde{U}\cong \widetilde{U}\times\wedge^{*}_{C}(\frac{n}{2}).$ Let $\{f_{1},\cdots,f_{n}\}$ be the orthonormal basis of
$\wedge^{*}_{C}(\frac{n}{2})$. Take a spin frame field $\sigma: \widetilde{U}\rightarrow Spin(M)$ such that
$\pi\sigma=\{\widetilde{e_{1}},\cdots, \widetilde{e_{n}}\}$ where $\pi: Spin(M)\rightarrow O(M)$ is a double covering, then
$\{[\sigma, f_{i}], 1\leq i\leq 4\}$ is an orthonormal frame of $S(TM)|_{\widetilde{U}}.$ In the following, since the global form $\Phi$
is independent of the choice of the local frame, so we can compute $\texttt{tr}_{S(TM)}$ in the frame $\{[\sigma, f_{i}], 1\leq i\leq 4\}$.
Let $\{E_{1},\cdots,E_{n}\}$ be the canonical basis of $R^{n}$ and
$c(E_{i})\in cl_{C}(n)\cong Hom(\wedge^{*}_{C}(\frac{n}{2}),\wedge^{*}_{C}(\frac{n}{2}))$ be the Clifford action. By \cite{Y}, then

\begin{equation}
c(\widetilde{e_{i}})=[(\sigma,c(E_{i}))]; \quad c(\widetilde{e_{i}})[(\sigma, f_{i})]=[\sigma,(c(E_{i}))f_{i}]; \quad
\frac{\partial}{\partial x_{i}}=[(\sigma,\frac{\partial}{\partial x_{i}})],
\end{equation}
then we have $\frac{\partial}{\partial x_{i}}c(\widetilde{e_{i}})=0$ in the above frame. By Lemma 2.2 in \cite{Wa3}, we have

\begin{lem}\label{le:32}
With the metric $\frac{1}{h(x_{n})}g^{\partial M}+\texttt{d}x _{n}^{2}$ on $M$ near the boundary
\begin{eqnarray}
&&  \partial x_{j}(|\xi|^{2}_{g^{M}})(x_{0})=0,
     \quad if~~ j<n ;\quad  =h'(0)|\xi'|^{2}_{g^{\partial M}},\quad  if~~ j=n.  \nonumber\\
&&  \partial x_{j}(c(\xi))(x_{0})=0,
     \quad if~~ j<n ;\quad =\partial x_{n}(c(\xi'))(x_{0}),\quad if~~ j=n.
\end{eqnarray}
where $\xi=\xi'+\xi_{n}\texttt{d}x_{n}$
\end{lem}
Then an application of Lemma 2.3 in \cite{Wa3} shows
\begin{lem}
\begin{eqnarray}
\sigma_{1}(D_{T})&=&\sigma_{1}(D_{T}^{*})=\sqrt{-1}c(\xi); \\
\sigma_{0}(D_{T})&=&-\frac{3}{4}h'(0)c(dx_n)
+\frac{1}{4}\sum_{i\neq s\neq t}A_{ist}
  c(\widetilde{e_{i}})c(\widetilde{e_{s}})c(\widetilde{e_{t}})  \nonumber\\
  &&+\frac{1}{4}\sum_{i, s, t}[-A_{iit}c(\widetilde{e_{t}})
  +A_{isi}c(\widetilde{e_{s}}) -A_{iss}c(\widetilde{e_{i}})+2A_{iii}c(\widetilde{e_{i}})], \\
\sigma_{0}(D_{T}^{*})&=&-\frac{3}{4}h'(0)c(dx_n)
 +\frac{1}{4}\sum_{i\neq s\neq t}A_{ist}
  c(\widetilde{e_{i}})c(\widetilde{e_{s}})c(\widetilde{e_{t}})  \nonumber\\
  &&-\frac{1}{4}\sum_{i, s, t}[-A_{iit}c(\widetilde{e_{t}})
  +A_{isi}c(\widetilde{e_{s}}) -A_{iss}c(\widetilde{e_{i}})+2A_{iii}c(\widetilde{e_{i}})].
\end{eqnarray}
\end{lem}

Now we can compute $\Phi$ (see formula (2.13) for definition of $\Phi$), since the sum is taken over $-r-\ell+k+j+|\alpha|=3,
 \ r, \ell\leq-1$, then we have the following five cases:

\textbf{Case a(I)}: \ $r=-1, \ \ell=-1, \ k=j=0, \ |\alpha|=1$

From (2.13), we have
\begin{equation}
\text{ Case \ a(\text{I})}=-\int_{|\xi'|=1}\int_{-\infty}^{+\infty}\sum_{|\alpha|=1}\text{trace}
[\partial_{\xi'}^{\alpha}\pi_{\xi_{n}}^{+}\sigma_{-1}((D_{T}^{*})^{-1})
\partial_{x'}^{\alpha}\partial_{\xi_{n}}\sigma_{-1}(D_{T}^{-1})](x_{0})\texttt{d}\xi_{n}\sigma(\xi')\texttt{d}x' .
\end{equation}
By Lemma 3.5, for $j<n$
\begin{equation}
\partial_{x_i}\sigma_{-1}(D_{T}^{-1})(x_0)=\partial_{x_i}\left(\frac{\sqrt{-1}c(\xi)}{|\xi|^2}\right)(x_0)=
\frac{\sqrt{-1}\partial_{x_i}[c(\xi)](x_0)}{|\xi|^2}
-\frac{\sqrt{-1}c(\xi)\partial_{x_i}(|\xi|^2)(x_0)}{|\xi|^4}=0,
\end{equation}
so Case a(I) vanishes.

\textbf{Case a(II)}: \ $r=-1, \ \ell=-1, \ k=|\alpha|=0, \ j=1$

From (2.13), we have
\begin{equation}
\text{ Case \ a(\text{II})}=-\frac{1}{2}\int_{|\xi'|=1}\int_{-\infty}^{+\infty}\text{trace}[\partial_{x_{n}}\pi_{\xi_{n}}^{+}
\sigma_{-1}((D_{T}^{*})^{-1})\partial_{\xi_{n}}^{2}\sigma_{-1}(D_{T}^{-1})](x_{0})\texttt{d}\xi_{n}\sigma(\xi')\texttt{d}x'.
\end{equation}
Similarly to (2.2.18) in \cite{Wa3}, we have
\begin{equation}
\partial_{x_{n}}\pi_{\xi_{n}}^{+}\sigma_{-1}((D_{T}^{*})^{-1})(x_{0})|_{|\xi'|=1}=\frac{\partial_{x_{n}}[c(\xi')](x_{0})}{2(\xi_{n}-i)}
+\sqrt{-1}h'(0)[\frac{ic(\xi')}{4(\xi_{n}-i)}+\frac{c(\xi')+ic(\texttt{d}x_{n})}{4(\xi_{n}-i)^{2}}];
\end{equation}
\begin{equation}
\partial_{\xi_{n}}^{2}\sigma_{-1}(D_{T}^{-1})=\sqrt{-1}(-\frac{6\xi_{n}c(\texttt{d}x_{n})+2c(\xi')}{|\xi|^{4}}
+\frac{8\xi_{n}^{2}c(\xi)}{|\xi|^{6}}).
\end{equation}
By the relation of the Clifford action and $\texttt{tr}AB=\texttt{tr}BA$, then
\begin{eqnarray}
&&\texttt{tr}[c(\xi')c(\texttt{d}x_{n})]=0; \ \texttt{tr}[c(\texttt{d}x_{n})^{2}]=-4;\ \texttt{tr}[c(\xi')^{2}](x_{0})|_{|\xi'|=1}=-4;\nonumber\\
&&\texttt{tr}[\partial_{x_{n}}[c(\xi')]c(\texttt{d}x_{n})]=0; \ \texttt{tr}[\partial_{x_{n}}c(\xi')\times c(\xi')](x_{0})|_{|\xi'|=1}=-2h'(0).
\end{eqnarray}
 For more trace expansions, we can see \cite{GS}.
Hence we conclude that
\begin{equation}
\text{trace}[\partial_{\xi'}^{\alpha}\pi_{\xi_{n}}^{+}\sigma_{-1}((D_{T}^{*})^{-1})\partial_{x'}^{\alpha}\partial_{\xi_{n}}\sigma_{-1}
(D_{T}^{-1})](x_{0})
=\frac{2ih'(0)}{(\xi_{n}-i)^{2}(\xi_{n}+i)^{3}}.
\end{equation}
Therefore
\begin{equation}
\text{ Case \ a(\text{II}) }=-\frac{3}{8}\pi h'(0)\Omega_{3}\texttt{d}x',
\end{equation}
where $\Omega_{3}$ is the canonical volume of $S^{3}.$

\textbf{Case a(III)}: \ $r=-1, \ \ell=-1, \ j=|\alpha|=0, \ k=1$

From (2.13), we have
\begin{equation}
\text{ Case \ a(\text{III}) }=-\frac{1}{2}\int_{|\xi'|=1}\int_{-\infty}^{+\infty}\text{trace}[\partial_{\xi_{n}}\pi_{\xi_{n}}^{+}
\sigma_{-1}((D_{T}^{*})^{-1})\partial_{\xi_{n}}\partial_{x_{n}}\sigma_{-1}(D_{T}^{-1})](x_{0})\texttt{d}\xi_{n}\sigma(\xi')\texttt{d}x' .
 \end{equation}
Similarly to (2.2.27) in \cite{Wa3}, we have
\begin{equation}
\partial_{\xi_{n}}\pi_{\xi_{n}}^{+}\sigma_{-1}((D^{*})^{-1})(x_{0})|_{|\xi'|=1}=-\frac{c(\xi')+ic(\texttt{d}x_{n})}{2(\xi_{n}-i)^{2}},
\end{equation}
and
\begin{equation}
\partial_{\xi_{n}}\partial_{x_{n}}\sigma_{-1}(D^{-1})(x_{0})|_{|\xi'|=1}
=-\sqrt{-1}h'(0)\Big[\frac{c(\texttt{d}x_{n})}{|\xi|^{4}}-4\xi_{n}\frac{c(\xi')+\xi_{n}c(\texttt{d}x_{n})}{|\xi|^{6}}\Big]
-\frac{2\sqrt{-1}\xi_{n}\partial_{x_{n}}c(\xi')(x_{0})}{|\xi|^{4}}.
\end{equation}
Combining (3.34) and (3.35), we obtain
\begin{equation}
\text{trace}[\partial_{\xi_{n}}\pi_{\xi_{n}}^{+}\sigma_{-1}(D_{T}^{-1})\partial_{\xi_{n}}\partial_{x_{n}}\sigma_{-1}(D_{T}^{-1})](x_{0})
=\frac{2h'(0)(i-2\xi_{n}-i\xi_{n}^{2})}{(\xi_{n}-i)^{4}(\xi_{n}+i)^{3}}.
\end{equation}
Then
\begin{equation}
\text{ Case \ a(\text{III}) }=\frac{3}{8}\pi h'(0)\Omega_{3}\texttt{d}x',
\end{equation}
where $\Omega_{3}$ is the canonical volume of $S^{3}.$
Thus the sum of Case a(\texttt{II}) and Case a(\texttt{III}) is zero.

\textbf{Case b}: \ $r=-2, \ \ell=-1, \ k=j=|\alpha|=0$

By (2.13), we get
\begin{equation}
\text{ Case \ b}=-i\int_{|\xi'|=1}\int_{-\infty}^{+\infty}\text{trace}[\pi_{\xi_{n}}^{+}\sigma_{-2}((D_{T}^{*})^{-1})
                 \partial_{\xi_{n}}\sigma_{-1}(D_{T}^{-1})](x_{0})\texttt{d}\xi_{n}\sigma(\xi')\texttt{d}x'.
\end{equation}
Then an application of Lemma 3.4 and Lemma 3.5 shows
\begin{eqnarray}
\sigma_{-2}((D_{T}^{*})^{-1})(x_{0})&=&\frac{c(\xi)\sigma_{0}(D_{T}^{*})(x_{0})c(\xi)}{|\xi|^{4}}+\frac{c(\xi)}{|\xi|^{6}}\sum_{j}c(dx_{j})
                    \Big[\partial_{x_{j}}(c(\xi))|\xi|^{2}-c(\xi)\partial_{x_{j}}(|\xi|^{2})\Big](x_{0})\nonumber\\
                   &=&\frac{c(\xi)\sigma_{0}(D_{T}^{*})(x_{0})c(\xi)}{|\xi|^{4}}+
                  + \frac{c(\xi)}{|\xi|^{6}}c(dx_{n})\Big[\partial x_{n}(c(\xi'))(x_{0})-c(\xi)h'(0)|\xi'|^{2}_{g^{\partial M}}\Big].
\end{eqnarray}
Hence in this case,
\begin{equation}
\pi_{\xi_{n}}^{+}\sigma_{-2}((D^{*})^{-1})(x_{0}):=A_{1}+A_{2},
\end{equation}
where
\begin{eqnarray}
A_{1}&=&
-\frac{h'(0)}{2}\left[\frac{c(dx_n)}{4i(\xi_n-i)}+\frac{c(dx_n)-ic(\xi')}{8(\xi_n-i)^2}
+\frac{3\xi_n-7i}{8(\xi_n-i)^3}[ic(\xi')-c(dx_n)]\right]  \nonumber\\
&&+\frac{-1}{4(\xi_n-i)^2}\Big[(2+i\xi_n)c(\xi')\alpha_0c(\xi')+i\xi_nc(dx_n)\alpha_0c(dx_n)+(2+i\xi_n)c(\xi')c(dx_n)\partial_{x_n}c(\xi')\nonumber\\
&&~~~~+ic(dx_n)\alpha_0c(\xi')+ic(\xi')\alpha_0c(dx_n)-i\partial_{x_n}c(\xi')\Big];\\
A_{2}&=&\frac{-1}{4(\xi_n-i)^2}\Big[(2+i\xi_n)c(\xi')\beta_0c(\xi')+i\xi_nc(dx_n)\beta_0c(dx_n)+ic(dx_n)\beta_0c(\xi')\nonumber\\
 &&+ic(\xi')\beta_0c(dx_n)\Big]
\end{eqnarray}
and
\begin{eqnarray}
\alpha_0&=&-\frac{3}{4}h'(0)c(dx_n),\\
\beta_0&=&\frac{1}{4}\sum_{i\neq s\neq t}A_{ist}
  c(\widetilde{e_{i}})c(\widetilde{e_{s}})c(\widetilde{e_{t}}) -\frac{1}{4}\sum_{i, s, t}\Big[-A_{iit}c(\widetilde{e_{t}})
  +A_{isi}c(\widetilde{e_{s}}) -A_{iss}c(\widetilde{e_{i}})+2A_{iii}c(\widetilde{e_{i}})\Big].
\end{eqnarray}
On the other hand,
\begin{equation}
\partial_{\xi_{n}}\sigma_{-1}(D_{T}^{-1})=\frac{-2 i \xi_{n}c(\xi')}{(1+\xi_{n}^{2})^{2}}+\frac{i(1- \xi_{n}^{2})c(\texttt{d}x_{n})}
{(1+\xi_{n}^{2})^{2}}.
\end{equation}
From (2.2.39), (2.2.41) and (2.2.42) in \cite{Wa3}, we have
\begin{equation}
-i\int_{|\xi'|=1}\int_{-\infty}^{+\infty}\text{trace}[A_{1}\times
                 \partial_{\xi_{n}}\sigma_{-1}(D_{T}^{-1})](x_{0})\texttt{d}\xi_{n}\sigma(\xi')\texttt{d}x'
                 =\frac{9}{8}\pi h'(0)\Omega_{3}\texttt{d}x'.
\end{equation}
By the relation of the Clifford action and $\texttt{tr}AB=\texttt{tr}BA$, then we have the equalities
\begin{equation}
\texttt{tr}[c(\widetilde{e_{i}})c(\texttt{d}x_{n})]=0, i<n; \ \texttt{tr}[c(\widetilde{e_{i}})c(\texttt{d}x_{n})]=-4, i=n.
\end{equation}
Combining (3.42), (3.45) and (3.47), we obtain
\begin{equation}
\text{trace}[A_{2}\times\partial_{\xi_{n}}\sigma_{-1}(D_{T}^{-1})](x_{0})=\frac{i c_0}{2(\xi_{n}+i)^{2}(\xi_{n}-i)},
\end{equation}
where
\begin{equation}
c_0= -2\sum_{i}A_{iin}.
\end{equation}
Hence from (3.42) and (3.45), we have
\begin{eqnarray}
&&-i\int_{|\xi'|=1}\int_{-\infty}^{+\infty}\text{trace}[A_{2}\times
                 \partial_{\xi_{n}}\sigma_{-1}(D_{T}^{-1})](x_{0})\texttt{d}\xi_{n}\sigma(\xi')\texttt{d}x'\nonumber\\
&=&-i\Omega_{3}\int_{\Gamma^{+}}\frac{i c_0}{2(\xi_{n}+i)^{2}(\xi_{n}-i)}\texttt{d}\xi_{n}\texttt{d}x' \nonumber\\
&=&-i 2\pi i \Omega_{3}[\frac{i c_0}{2(\xi_{n}+i)^{2}}]^{(1)}|_{\xi_{n}=i}\texttt{d}x'  \nonumber\\
&=&\frac{1}{4}\pi c_0\Omega_{3}\texttt{d}x'.
\end{eqnarray}
Combining (3.46) and (3.50), we have
\begin{equation}
\texttt{case}\ b=\Big[\frac{9}{8}h'(0)-\frac{1}{2}\sum_{i}A_{iin}\Big]\pi \Omega_{3}\texttt{d}x'.
\end{equation}

\textbf{Case c}: \ $r=-1, \ \ell=-2, \ k=j=|\alpha|=0$

From (2.13), we have
\begin{equation}
 \text{case \ c}=-i\int_{|\xi'|=1}\int_{-\infty}^{+\infty}\text{trace}[\pi_{\xi_{n}}^{+}\sigma_{-1}((D_{T}^{*})^{-1})
 \partial_{\xi_{n}}\sigma_{-2}(D_{T}^{-1})](x_{0})\texttt{d}\xi_{n}\sigma(\xi')\texttt{d}x'  .
 \end{equation}
Then an application of Lemma 3.4 shows
\begin{equation}
\pi_{\xi_{n}}^{+}\sigma_{-1}((D_{T}^{*})^{-1})=\frac{c(\xi')+ic(\texttt{d}x_{n})}{2(\xi_{n}-i)}.
\end{equation}
By Lemma 3.5 and Lemma 3.6, we have
\begin{eqnarray}
\sigma_{-2}((D_{T})^{-1})(x_{0})&=&\frac{c(\xi)\sigma_{0}(D_{T})(x_{0})c(\xi)}{|\xi|^{4}}+\frac{c(\xi)}{|\xi|^{6}}\sum_{j}c(dx_{j})
                    \Big[\partial_{x_{j}}(c(\xi))|\xi|^{2}-c(\xi)\partial_{x_{j}}(|\xi|^{2})\Big](x_{0})\nonumber\\
                   &=&\frac{c(\xi)\sigma_{0}(D_{T})(x_{0})c(\xi)}{|\xi|^{4}}+
                  + \frac{c(\xi)}{|\xi|^{6}}c(dx_{n})\Big[\partial x_{n}(c(\xi'))(x_{0})-c(\xi)h_{x_{n}}'(0)|\xi'|^{2}_{g^{\partial M}}\Big].
\end{eqnarray}
Hence in this case,
\begin{equation}
 \partial_{\xi_{n}}\sigma_{-2}(D_{T}^{-1})(x_{0}):=B_{1}+B_{2},
\end{equation}
where
\begin{eqnarray}
B_{1}&=&\frac{1}{(1+\xi_n^2)^3}\Big[(2\xi_n-2\xi_n^3)c(dx_n)\alpha_0c(dx_n)+(1-3\xi_n^2)c(dx_n)\alpha_0c(\xi')\nonumber\\
    &&+ (1-3\xi_n^2)c(\xi')\alpha_0c(dx_n)-4\xi_nc(\xi')\alpha_0 c(\xi')+(3\xi_n^2-1)\partial_{x_n}c(\xi')
    -4\xi_nc(\xi')c(dx_n)\partial_{x_n}c(\xi')\nonumber\\
    &&+2h'(0)c(\xi')+2h'(0)\xi_nc(dx_n)\Big]
    +6\xi_nh'(0)\frac{c(\xi)c(dx_n)c(\xi)}{(1+\xi^2_n)^4};\\
B_{2}&=&\frac{1}{(1+\xi_n^2)^3}\Big[(2\xi_n-2\xi_n^3)c(dx_n)\beta_1c(dx_n)+(1-3\xi_n^2)c(dx_n)\beta_1c(\xi')\nonumber\\
    &&+(1-3\xi_n^2)c(\xi')\beta_1c(dx_n)-4\xi_nc(\xi')\beta_1c(\xi')\Big]
\end{eqnarray}
and
\begin{eqnarray}
\beta_1&=&\frac{1}{4}\sum_{i\neq s\neq t}A_{ist}
  c(\widetilde{e_{i}})c(\widetilde{e_{s}})c(\widetilde{e_{t}})+\frac{1}{4}\sum_{i, s, t}[-A_{iit}c(\widetilde{e_{t}})
  +A_{isi}c(\widetilde{e_{s}}) -A_{iss}c(\widetilde{e_{i}})+2A_{iii}c(\widetilde{e_{i}})].
\end{eqnarray}
Then similarly to computations of the (3.50), we have
\begin{equation}
-i\int_{|\xi'|=1}\int_{-\infty}^{+\infty}\text{trace}[\pi_{\xi_{n}}^{+}\sigma_{-1}((D_{T}^{*})^{-1})\times B_{1}]
               (x_{0})\texttt{d}\xi_{n}\sigma(\xi')\texttt{d}x'
                 =-\frac{9}{8}\pi h'(0)\Omega_{3}\texttt{d}x'.
\end{equation}
From (3.53) and (3.57) we obtain
\begin{equation}
\text{trace}[\pi_{\xi_{n}}^{+}\sigma_{-1}((D_{T}^{*})^{-1})\times B_{2}](x_{0})
=\frac{-i \tilde{c}_0}{(\xi_{n}+i)^{3}(\xi_{n}-i)},
\end{equation}
where
\begin{equation}
\tilde{c}_0= 2\sum_{i}A_{iin}.
\end{equation}
Then
\begin{eqnarray}
&&-i\int_{|\xi'|=1}\int_{-\infty}^{+\infty}\text{trace}[\pi_{\xi_{n}}^{+}\sigma_{-1}((D_{T}^{*})^{-1})\times B_{2}]
                  (x_{0})\texttt{d}\xi_{n}\sigma(\xi')\texttt{d}x'\nonumber\\
&=&-i\Omega_{3}\int_{\Gamma^{+}}\frac{i \tilde{c}_0}{(\xi_{n}+i)^{3}(\xi_{n}-i)}\texttt{d}\xi_{n}\texttt{d}x' \nonumber\\
&=& 2\pi i \Omega_{3}\frac{ \tilde{c}_0}{(\xi_{n}+i)^{3}}|_{\xi_{n}=i}\texttt{d}x'  \nonumber\\
&=&-\frac{1}{4}\pi \tilde{c}_0\Omega_{3}\texttt{d}x'.
\end{eqnarray}
Combining (3.59) and (3.62), we have
\begin{equation}
\texttt{case}\ c=\Big[-\frac{9}{8}h'(0)-\frac{1}{2}\sum_{i}A_{iin}\Big]\pi \Omega_{3}\texttt{d}x'.
\end{equation}
Now $\Phi$  is the sum of the \textbf{case (a, b, c)}, so
\begin{equation}
\sum \textbf{case a, b , c}=-\sum_{i}A_{iin}\pi \Omega_{3}\texttt{d}x'.
\end{equation}
Hence we conclude that

\begin{thm}
 Let M be a 4-dimensional compact manifold  with the boundary $\partial M$ and $\tilde{\nabla}$ be an orthogonal
connection with torsion. Then we get the volumes  associated to $D^{*}D$,
 \begin{equation}
Vol_{4}^{(1, 1)}=-\frac{1}{48\pi^{2}}\int_{M}\tilde{R}(x)dx-\int_{\partial_{M}}\sum_{i}A_{iin}\pi \Omega_{3}\texttt{d}x',
\end{equation}
where $\tilde{R}=R+18div(V)-54|V|^{2}-9\parallel T\parallel^{2}$ and $\int_{M}div(V) dVol_{M}=-\int_{\partial_{M}}g(n,V)dVol_{\partial_{M}}$.
\end{thm}
\section{The gravitational action for 4-dimensional manifolds with boundary}
\label{4}
Firstly, we recall the Einstein-Hilbert action with torsion for manifolds with boundary (see \cite{Wa3} or \cite{Wa4}),
 \begin{equation}
 I_{Gr}=\frac{1}{16\pi}\int_{M}\tilde{R}\texttt{d}vol_{M}+2\int_{\partial_{M}}\tilde{K}\texttt{d}vol_{\partial_{M}}:=I_{Gr, i}+I_{Gr, b},
 \end{equation}
where $\tilde{R}=R+18div(V)-54|V|^{2}-9\parallel T\parallel^{2}$ be the scalar curvature of this orthogonal connection without the Cartan
type torsion $S$.
And
\begin{equation}
\tilde{K}=K+\sum_{i}A_{iin}; \ K=\sum_{1\leq i,j\leq n-1}K_{i,j}g_{\partial_{M}}^{i,j},
\end{equation}
where $K_{i,j}$ is the second fundamental form, or extrinsic curvature. Take the metric in Section 2, and by Lemma A.2 in \cite{Wa3},
  for $n=4$, we assume the manifold approach  the boundary $\partial_{M}$ is flat, then
  \begin{equation}
\tilde{K}(x_{0})=\sum_{i}A_{iin} , \ K(x_{0})=0.
\end{equation}

Let
 \begin{equation}
\widetilde{{\rm Wres}}[\pi^+(D_{T}^{*})^{-1}\circ\pi^+D_{T}^{-1}]
=\widetilde{{\rm Wres}}_{i}[\pi^+(D_{T}^{*})^{-1}\circ\pi^+D_{T}^{-1}]+\widetilde{{\rm Wres}}_{b}[\pi^+(D_{T}^{*})^{-1}\circ\pi^+D_{T}^{-1}],
\end{equation}
where
 \begin{equation}
\widetilde{{\rm Wres}}_{i}[\pi^+(D_{T}^{*})^{-1}\circ\pi^+D_{T}^{-1}]
=\int_M\int_{|\xi|=1}{\rm trace}_{S(TM)}[\sigma_{-4}((D_{T}^{*})^{-1} \circ D_{T}^{-1})]\sigma(\xi)dx
\end{equation}
and
\begin{eqnarray}
&&\widetilde{{\rm Wres}}_{b}[\pi^+((D_{T}^{*})^{-1}\circ\pi^+D_{T}^{-1}]\nonumber\\
&=& \int_{\partial M}\int_{|\xi'|=1}\int^{+\infty}_{-\infty}\sum^{\infty}_{j, k=0}\sum\frac{(-i)^{|\alpha|+j+k+1}}{\alpha!(j+k+1)!}
\times {\rm trace}_{S(TM)}[\partial^j_{x_n}\partial^\alpha_{\xi'}\partial^k_{\xi_n}\sigma^+_{r}(((D_{T}^{*})^{-1})(x',0,\xi',\xi_n)
\nonumber\\
&&\times\partial^\alpha_{x'}\partial^{j+1}_{\xi_n}\partial^k_{x_n}\sigma_{l}(D_{T}^{-1})(x',0,\xi',\xi_n)]d\xi_n\sigma(\xi')dx'
\end{eqnarray}
denote the interior term and boundary term of $\widetilde{{\rm Wres}}[\pi^+(D_{T}^{*})^{-1}\circ\pi^+D_{T}^{-1}]$.

Combining (3.65), (4.1) and (4.4), we obtain
\begin{thm}
 Let M be a 4-dimensional compact manifold  with the boundary $\partial M$ and $\tilde{\nabla}$ be an orthogonal
connection with torsion. Then we get the volumes  associated to $D^{*}D$,
\begin{eqnarray}
&& I_{\rm {Gr,i}}=-3\pi \widetilde{{\rm Wres}}_{i}[\pi^+(D_{T}^{*})^{-1}\circ\pi^+D_{T}^{-1}]; \nonumber\\
&&I_{\rm {Gr,b}}=\frac{-2}{\pi\Omega_3 }\widetilde{{\rm Wres}}_{b}[\pi^+(D_{T}^{*})^{-1}\circ\pi^+D_{T}^{-1}].
\end{eqnarray}
\end{thm}

\section{A Kastler-Kalau-Walze type theorem for $6$-dimensional spin manifolds with boundary
associated to $(D_{T}^{*})^{2}$ and $D_{T}^{2}$ }
 In this section, We compute the lower dimensional volume ${\rm Vol}^{(2,2)}_6$ for $6$-dimensional spin manifolds with
boundary of  metric $ g^M=\frac{1}{h(x_n)}g^{\partial M}+dx_n^2$
 and get a Kastler-Kalau-Walze type theorem in this case.

Firstly, we compute $\int_{\partial M}\Phi$ in this case. By Lemma 1 in \cite{Wa4} , we have

\begin{lem}
\begin{eqnarray}
&&\sigma_{-2}((D_{T}^{*})^{-2})=|\xi|^{-2}; \\
&&\sigma_{-2}(D_{T}^{-2})=|\xi|^{-2}; \\
&&\sigma_{-3}((D_{T}^{*})^{-2})=-\sqrt{-1}|\xi|^{-4}\xi_k(\tilde{\Gamma}^k-2\hat{\delta}^k)-\sqrt{-1}|\xi|^{-6}2\xi^j\xi_\alpha\xi_\beta
\partial_jg^{\alpha\beta}-2\sqrt{-1}|\xi|^{-4}(u-v)c(\xi);\\
&&\sigma_{-3}(D_{T}^{-2})=-\sqrt{-1}|\xi|^{-4}\xi_k(\tilde{\Gamma}^k-2\check{\delta}^k)-\sqrt{-1}|\xi|^{-6}2\xi^j\xi_\alpha\xi_\beta
\partial_jg^{\alpha\beta}-2\sqrt{-1}|\xi|^{-4}(u+v)c(\xi).
\end{eqnarray}
\end{lem}

 Now we can compute $\Phi$ (see formula (2.13) for the definition of $\Phi$), since the sum is taken over $
-r-l+k+j+|\alpha|=5,~~r,~l\leq -2,$ then we have the following five cases:

{\bf case $\hat{a}$)~I)}~$r=-2,~l=-2,~k=j=0,~|\alpha|=1$

From (2.13) we have
 \begin{equation}
{\rm case~\hat{a})~I)}=-\int_{|\xi'|=1}\int^{+\infty}_{-\infty}\sum_{|\alpha|=1}
 {\rm trace}[\partial^\alpha_{\xi'}\pi^+_{\xi_n}\sigma_{-2}((D_{T}^{*})^{-2})\times
 \partial^\alpha_{x'}\partial_{\xi_n}\sigma_{-2}(D_{T}^{-2})](x_0)d\xi_n\sigma(\xi')dx'.
\end{equation}
By Lemma 3.5, for $i<n$, then
 \begin{equation}
\partial_{x_i}\sigma_{-2}(D_{T}^{-2})(x_0)=\partial_{x_i}\left(\frac{1}{|\xi|^2}\right)(x_0)
=-\frac{\partial_{x_i}(|\xi|^2)(x_0)}{|\xi|^4}\nonumber\\
=0.
\end{equation}
Then case $\hat{a}$) I) vanishes.

 {\bf case $\hat{a}$)~II)}~$r=-1,~l=-1, ~k=|\alpha|=0,~j=1$

From (2.13) we have
 \begin{equation}
{\rm case~ \hat{a})~II)}=-\frac{1}{2}\int_{|\xi'|=1}\int^{+\infty}_{-\infty} {\rm
trace} [\partial_{x_n}\pi^+_{\xi_n}\sigma_{-2}((D_{T}^{*})^{-2})\times
\partial_{\xi_n}^2\sigma_{-2}(D_{T}^{-2})](x_0)d\xi_n\sigma(\xi')dx'.
\end{equation}
By Lemma 3.5 and Lemma 5.1, we have
 \begin{equation}
\partial^2_{\xi_n}\sigma_{-2}(D_{T}^{-2})(x_0)=\partial^2_{\xi_n}(\frac{1}{|\xi|^2})(x_0)=\frac{-2+6\xi_n^{2}}{(1+\xi_n^{2})^3},
\end{equation}
and
 \begin{equation}
\partial_{x_n}\sigma_{-2}((D_{T}^{*})^{-2})(x_0)=\frac{-h'(0)}{(1+\xi_n^{2})^3}.
\end{equation}
Then
 \begin{equation}
\pi^+_{\xi_n}\left[\partial_{x_n}\sigma_{-2}((D_{T}^{*})^{-2})\right](x_0)|_{|\xi'|=1}=\frac{(i\xi_n^{2}+2)h'(0)}{4(\xi_n-i)^{2}}.
\end{equation}
Combining (5.7) and (5.9), we obtain
\begin{eqnarray}
 &&\int^{+\infty}_{-\infty}
\frac{(i\xi_n^{2}+2)h'(0)}{4(\xi_n-i)^{2}}\times
 \frac{-2+6\xi_n^{2}}{(1+\xi_n^{2})^3} d\xi_n\nonumber\\
&=&-\frac{1}{2}\int_{\Gamma^+} \frac{(3\xi_n^{2}-1)\big(-2h'(0)-i\xi_n h'(0)\big)}
{(\xi_n-i)^5(\xi_n+i)^3}d\xi_ndx'\nonumber\\
&=&-\frac{1}{2}\pi i\Big[\frac{(3\xi_n^{2}-1)\big(-2h'(0)-i\xi_nh'(0)\big)}
{(\xi_n+i)^3}\Big]^{(4)}|_{\xi_n=i}dx'\nonumber\\
&=&\frac{5h'(0)}{32}.
\end{eqnarray}
Since $n=6$, ${\rm tr}_{S(TM)}[{\rm id}]={\rm dim}(\wedge^*(3))=8.$
Combining (5.6) and (5.10), we have
 \begin{equation}
{\rm case~ \hat{a})~II)}=-\frac{5h'(0)}{8} \Omega_4dx'.
\end{equation}
 where $\Omega_4$ is the canonical volume of $S^4$.

{\bf case $\hat{a}$)~III)}~$r=-2,~l=-2,~j=|\alpha|=0,~k=1$\\

From (2.13)  and an integration by parts, we get
\begin{eqnarray}
{\rm case~ \hat{a})~III)}&=&-\frac{1}{2}\int_{|\xi'|=1}\int^{+\infty}_{-\infty}
{\rm trace} [\partial_{\xi_n}\pi^+_{\xi_n}\sigma_{-2}((D_{T}^{*})^{-2})\times
\partial_{\xi_n}\partial_{x_n}\sigma_{-2}(D_{T}^{-2})](x_0)d\xi_n\sigma(\xi')dx'\nonumber\\
&=&\frac{1}{2}\int_{|\xi'|=1}\int^{+\infty}_{-\infty} {\rm trace}
[\partial_{\xi_n}^2\pi^+_{\xi_n}\sigma_{-2}((D_{T}^{*})^{-2})\times
\partial_{x_n}\sigma_{-2}(D_{T}^{-2})](x_0)d\xi_n\sigma(\xi')dx'.
\end{eqnarray}
By Lemma 3.5 and Lemma 5.1, we have
 \begin{equation}
\partial_{\xi_n}^2\pi^+_{\xi_n}\sigma_{-2}((D_{T}^{*})^{-2})(x_0)|_{|\xi'|=1}=\frac{-i}{(\xi_n-i)^3}.
\end{equation}
Substituting (5.8) and (5.13) into (5.12), one sees that
\begin{eqnarray}
{\rm {\bf case~\hat{a})~III)}}&=&
\frac{1}{2}\int_{|\xi'|=1}\int^{+\infty}_{-\infty}\frac{8ih'(0)}{(\xi_n-i)^5(\xi_n+i)^2}d\xi_n\sigma(\xi')dx' \nonumber\\
&=&\frac{5h'(0)}{8}\pi \Omega_4dx'.
\end{eqnarray}

{\bf case $\hat{b}$)}~$r=-2,~l=-3,~k=j=|\alpha|=0$\\

From (2.13)  and an integration by parts, we get
\begin{eqnarray}
{\rm case~ \hat{b})}&=&-i\int_{|\xi'|=1}\int^{+\infty}_{-\infty}
{\rm trace} [\pi^+_{\xi_n}\sigma_{-2}((D_{T}^{*})^{-2})\times
\partial_{\xi_n}\sigma_{-3}(D_{T}^{-2})](x_0)d\xi_n\sigma(\xi')dx'\nonumber\\
&=&i\int_{|\xi'|=1}\int^{+\infty}_{-\infty}
{\rm trace} [\partial_{\xi_n}\pi^+_{\xi_n}\sigma_{-2}((D_{T}^{*})^{-2})\times
\sigma_{-3}(D_{T}^{-2})](x_0)d\xi_n\sigma(\xi')dx'.
\end{eqnarray}
By Lemma 5.1, we have
 \begin{equation}
\partial_{\xi_n}\pi_{\xi_n}^+\sigma_{-2}((D_{T}^{*})^{-2})(x_0)|_{|\xi'|=1}=\frac{i}{2(\xi_n-i)^2}.
\end{equation}

In the normal coordinate, $g^{ij}(x_0)=\delta_i^j$ and $\partial_{x_j}(g^{\alpha\beta})(x_0)=0,$ {\rm if
}$j<n;~=h'(0) \delta^\alpha_\beta,~{\rm if }~j=n.$ So by Lemma A.2 in \cite{Wa3},
we have $\Gamma^n(x_0)=\frac{5}{2}h'(0)$ and
$\Gamma^k(x_0)=0$ for $k<n$. By the definition of $\delta^k$ and Lemma 2.3 in \cite{Wa3}, we have $\delta^n(x_0)=0$ and
$\delta^k=\frac{1}{4}h'(0)c(\widetilde{e_k})c(\widetilde{e_n})$ for $k<n$. So
\begin{eqnarray}
&&\sigma_{-3}(D^{-2})(x_0)|_{|\xi'|=1}\nonumber\\
&=&-\sqrt{-1}|\xi|^{-4}\xi_k(\tilde{\Gamma}^k-2\delta^k)(x_0)|_{|\xi'|=1}-\sqrt{-1}|\xi|^{-6}2\xi^j\xi_\alpha\xi_\beta
\partial_jg^{\alpha\beta}(x_0)|_{|\xi'|=1}-2\sqrt{-1}|\xi|^{-4}(u+v)c(\xi)\nonumber\\
&=&\frac{-i}{(1+\xi_n^2)^2}\Big(-\frac{1}{2}h'(0)\sum_{k<n}\xi_k
c(\widetilde{e_k})c(\widetilde{e_n})+\xi_n\frac{5}{2}h'(0)\Big)
-\frac{2i\xi_n h'(0)}{(1+\xi_n^2)^3}-2\sqrt{-1}|\xi|^{-4}(u+v)c(\xi).
\end{eqnarray}
We note that $\int_{|\xi'|=1}\xi_1\cdots\xi_{2q+1}\sigma(\xi')=0$ and
$\textbf{tr}[v\times c(dx_{n})]=4\sum_{i}A_{iin}$ . Then
\begin{eqnarray}
{\bf case~ \hat{b})}
&=&i\int_{|\xi'|=1}\int^{+\infty}_{-\infty}\textbf{tr} [\partial_{\xi_n}\pi^+_{\xi_n}\sigma_{-2}((D_{T}^{*})^{-2})\times
   \sigma_{-3}(D_{T}^{-2})](x_0)d\xi_n\sigma(\xi')dx' \nonumber\\
 &=& -\frac{15h'(0)}{8}\pi \Omega_4dx'+i\int_{|\xi'|=1}\int^{+\infty}_{-\infty}
  \textbf{tr}\Big[\frac{i}{2(\xi_n-i)^2}\times
   \Big(-2\sqrt{-1}|\xi|^{-4}(u+v)c(\xi)\Big)\Big](x_0)d\xi_n\sigma(\xi')dx'\nonumber\\
&=&-\frac{15h'(0)}{8}\pi \Omega_4dx'+i\int_{|\xi'|=1}\int^{+\infty}_{-\infty}
 \frac{1}{(\xi_n-i)^2(1+\xi_n^{2})^{2}} \textbf{tr}\Big[
   (u+v)c(\xi)\Big](x_0)d\xi_n\sigma(\xi')dx'\nonumber\\
 &=&-\frac{15h'(0)}{8}\pi \Omega_4dx'+4i\sum_{i}A_{iin}
 \int_{|\xi'|=1}\int^{+\infty}_{-\infty} \frac{\xi_n}{(\xi_n-i)^2(1+\xi_n^{2})^{2}}d\xi_n\sigma(\xi')dx'\nonumber\\
 &=&\Big(-\frac{15h'(0)}{8}-\frac{1}{2}\sum_{i}A_{iin}\Big)\pi \Omega_4dx'.
\end{eqnarray}

 {\bf  case $\hat{c}$)}~$r=-3,~l=-2,~k=j=|\alpha|=0$\\

From (2.13) we have
 \begin{equation}
{\rm {\bf case~ \hat{c})}}=-i\int_{|\xi'|=1}\int^{+\infty}_{-\infty}{\rm trace} [\pi^+_{\xi_n}\sigma_{-3}((D_{T}^{*})^{-2})\times
\partial_{\xi_n}\sigma_{-2}(D_{T}^{-2})](x_0)d\xi_n\sigma(\xi')dx'.
\end{equation}
By (24) in \cite{Wa4}, we have
 \begin{equation}
{\rm {\bf case~ \hat{c})}}={\rm {\bf case~ \hat{b})}}-i\int_{|\xi'|=1}\int^{+\infty}_{-\infty}{\rm
tr}[\partial_{\xi_n}\sigma_{-2}(D_{T}^{-2})\times
\sigma_{-3}((D_{T}^{*})^{-2})]d\xi_n\sigma(\xi')dx'.
\end{equation}
Then an application of Lemma 5.1 shows
 \begin{equation}
\partial_{\xi_n}\sigma_{-2}(D_{T}^{-2})(x_0)=\frac{-2\xi_n}{(1+\xi_n^{2})^2}.
\end{equation}
Combining (5.3) and (5.21), we obtain
\begin{eqnarray}
&&-i\int_{|\xi'|=1}\int^{+\infty}_{-\infty}{\rm tr}[\partial_{\xi_n}\sigma_{-2}(D_{T}^{-2})\times
\sigma_{-3}((D_{T}^{*})^{-2})]d\xi_n\sigma(\xi')dx'\nonumber\\
 &=&\frac{15h'(0)}{4}\pi \Omega_4dx'-i\int_{|\xi'|=1}\int^{+\infty}_{-\infty}
  \textbf{tr}\Big[\frac{-2\xi_n}{(1+\xi_n^{2})^2}\times
   \Big(-2\sqrt{-1}|\xi|^{-4}(u-v)c(\xi)\Big)\Big](x_0)d\xi_n\sigma(\xi')dx'\nonumber\\
&=&\frac{15h'(0)}{4}\pi \Omega_4dx'-i\int_{|\xi'|=1}\int^{+\infty}_{-\infty}
 \frac{4i\xi_n}{(1+\xi_n^{2})^{4}} \textbf{tr}\Big[
   (u-v)c(\xi)\Big](x_0)d\xi_n\sigma(\xi')dx'\nonumber\\
 &=&\frac{15h'(0)}{4}\pi \Omega_4dx'-\sum_{i}A_{iin}
 \int_{|\xi'|=1}\int^{+\infty}_{-\infty} \frac{16\xi_n^{2}}{(\xi_n-i)^4(\xi_n+i)^{4}}d\xi_n\sigma(\xi')dx'\nonumber\\
 &=&\Big(\frac{15h'(0)}{4}-\sum_{i}A_{iin}\Big)\pi \Omega_4dx'.
\end{eqnarray}
From (5.18) and (5.22), we have
  \begin{equation}
{\bf case~ \hat{c})}=\Big(\frac{15h'(0)}{8}-\frac{3}{2}\sum_{i}A_{iin}\Big)\pi \Omega_4dx'.
\end{equation}
Since $\Phi$ is the sum of the cases $\hat{a}$), $\hat{b}$) and $\hat{c}$), so $\Phi=-2\sum_{i}A_{iin}\pi \Omega_4dx'$. Hence we conclude that

\begin{thm}
Let  $M$ be a $6$-dimensional compact spin manifold with the boundary $\partial M$ and $\tilde{\nabla}$ be an orthogonal
connection with torsion. Then we get the volumes  associated to $D_{T}^{*}D_{T}$ with torsion on $\widehat{M}$
 \begin{equation}
\widetilde{Wres}_{b}(\pi^{+}(D_{T}^{*})^{-2}\circ\pi^{+}(D_{T})^{-2})=-2\int_{\partial_{M}}\sum_{i}A_{iin}\pi \Omega_4dx';
\end{equation}
when $V=0$,
  \begin{equation}
Wres(\pi^{+}(D_{T}^{*})^{-2}\circ\pi^{+}(D_{T})^{-2})=-\frac{1}{48\pi^{2}}\int_{M}(R-9\parallel T\parallel^{2})(x)dx.
\end{equation}
 \end{thm}

\section{The Kastler-Kalau-Walze theorem  for $4$-dimensional spin manifolds with boundary about Dirac operator $P^{+}D_{T}^{*}D_{T}$ }

Next we consider the volume form $\omega^{g}$  acting on the spinor bundle $\Sigma M$. Setting $P = P^{+}=\frac{1}{2}(id_{\Sigma} +\omega^{g})$
 we have a parallel field of orthogonal projections. If we now calculate the Seeley-deWitt coefficients of $H^{+ }= P^{+}D^{*}D$,
we obtain relations to Loop Quantum Gravity. The Holst term for the modified connection $\tilde{\nabla}$ is the 4-form
 \begin{equation}
\tilde{C}_{H}=18(dT-\langle T, *V\rangle\omega^{g}),
\end{equation}
where $\omega^{g}=e_{1}^{*}\wedge e_{2}^{*}\wedge e_{3}^{*}\wedge e_{4}^{*}$, see Remark 3.2 and Proposition 3.3 in \cite{PS1}.

\begin{thm}\cite{PS1}
 Let M be a 4-dimensional compact manifold  without boundary  and let $\tilde{R}$ be the scalar curvature
of the modified connection $\tilde{\nabla}$. Then we get the volumes  associated to $P^{+}D_{T}^{*}D_{T}$,
 \begin{equation}
Wres(P^{+}(D_{T}^{*}D_{T})^{-1})=-\frac{1}{96\pi^{2}}\int_{M}(\tilde{R}\omega^{g}+\tilde{C}_{H}),
\end{equation}
where $\tilde{R}=R+18div(V)-54|V|^{2}-9\parallel T\parallel^{2}$ be the scalar curvature of the modified connection $\tilde{\nabla}$.
\end{thm}

\begin{thm}\cite{PS1}
 Let M be a 6-dimensional compact manifold  without boundary and let $\tilde{R}$ be the scalar curvature
of the modified connection $\tilde{\nabla}$. Then we get the volumes  associated to $P^{+}D_{T}^{*}D_{T}$,
\begin{eqnarray}
Wres(P^{+}(D_{T}^{*}D_{T})^{-1})&=&\frac{11}{1440}\mathcal{X}(M)-\frac{1}{96}p_{1}(M)-\frac{1}{640\pi^{2}}\int_{M}\parallel C^{g}\parallel^{2}dx
  \nonumber\\
  &&-\frac{3}{64\pi^{2}}\int_{M}\Big(\parallel \delta T\parallel^{2}+\parallel d(V)\parallel^{2}\big)dx+ \frac{1}{1152\pi^{2}}\int_{M} \tilde{R}\tilde{C}_{H}.
\end{eqnarray}
where $C^{g}$ denote the Weyl curvature of the Levi-Civita connection,  and $p_{1}(M)$ denote the first Pontryagin class of $M$ and
 $\mathcal{X}(M)$ denote the Euler characteristics of $M$.

\end{thm}

Now we can compute $\Phi$ (see formula (2.13) for definition of $\Phi$), since the sum is taken over $-r-\ell+k+j+|\alpha|=3,
 \ r, \ell\leq-1$, then we have the following five cases:

\textbf{Case $\bar{a})$ I)}: \ $r=-1, \ \ell=-1, \ k=j=0, \ |\alpha|=1$

By (2.13), we get
\begin{eqnarray}
{\rm Case~\bar{a})~I)}&=&-\int_{|\xi'|=1}\int_{-\infty}^{+\infty}\sum_{|\alpha|=1}\text{trace}
        [\partial_{\xi'}^{\alpha}\pi_{\xi_{n}}^{+}P^{+}\sigma_{-1}((D_{T}^{*})^{-1})
       \partial_{x'}^{\alpha}\partial_{\xi_{n}}\sigma_{-1}(D_{T}^{-1})](x_{0})\texttt{d}\xi_{n}\sigma(\xi')\texttt{d}x' \nonumber\\
&=&-\frac{1}{2}\int_{|\xi'|=1}\int_{-\infty}^{+\infty}\sum_{|\alpha|=1}\text{trace}
        [\partial_{\xi'}^{\alpha}\pi_{\xi_{n}}^{+}\sigma_{-1}((D_{T}^{*})^{-1})
       \partial_{x'}^{\alpha}\partial_{\xi_{n}}\sigma_{-1}(D_{T}^{-1})](x_{0})\texttt{d}\xi_{n}\sigma(\xi')\texttt{d}x' \nonumber\\
&&-\frac{1}{2}\int_{|\xi'|=1}\int_{-\infty}^{+\infty}\sum_{|\alpha|=1}\text{trace}
        [\partial_{\xi'}^{\alpha}\pi_{\xi_{n}}^{+}\omega^{g}\sigma_{-1}((D_{T}^{*})^{-1})
       \partial_{x'}^{\alpha}\partial_{\xi_{n}}\sigma_{-1}(D_{T}^{-1})](x_{0})\texttt{d}\xi_{n}\sigma(\xi')\texttt{d}x'\nonumber\\
\end{eqnarray}
By Lemma 3.5, for $j<n$
\begin{equation}
\partial_{x_i}\sigma_{-1}(D_{T}^{-1})(x_0)=\partial_{x_i}\left(\frac{\sqrt{-1}c(\xi)}{|\xi|^2}\right)(x_0)=
\frac{\sqrt{-1}\partial_{x_i}[c(\xi)](x_0)}{|\xi|^2}
-\frac{\sqrt{-1}c(\xi)\partial_{x_i}(|\xi|^2)(x_0)}{|\xi|^4}=0,
\end{equation}
so ${\rm Case~\bar{a})~I)}$ vanishes.

\textbf{Case $\bar{a}$ ) II)}: \ $r=-1, \ \ell=-1, \ k=|\alpha|=0, \ j=1$

From (2.13), we have
\begin{eqnarray}
{\rm Case~\bar{a})~II)}&=&-\frac{1}{2}\int_{|\xi'|=1}\int_{-\infty}^{+\infty}\text{trace}[\partial_{x_{n}}\pi_{\xi_{n}}^{+}
   P^{+}\sigma_{-1}((D_{T}^{*})^{-1})\partial_{\xi_{n}}^{2}\sigma_{-1}(D_{T}^{-1})](x_{0})\texttt{d}\xi_{n}\sigma(\xi')\texttt{d}x' \nonumber\\
&=&-\frac{1}{4}\int_{|\xi'|=1}\int_{-\infty}^{+\infty}\text{trace}[\partial_{x_{n}}\pi_{\xi_{n}}^{+}
\sigma_{-1}((D_{T}^{*})^{-1})\partial_{\xi_{n}}^{2}\sigma_{-1}(D_{T}^{-1})](x_{0})\texttt{d}\xi_{n}\sigma(\xi')\texttt{d}x' \nonumber\\
&&-\frac{1}{4}\int_{|\xi'|=1}\int_{-\infty}^{+\infty}\text{trace}[\partial_{x_{n}}\pi_{\xi_{n}}^{+}
   \omega^{g}\sigma_{-1}((D_{T}^{*})^{-1})\partial_{\xi_{n}}^{2}\sigma_{-1}(D_{T}^{-1})](x_{0})\texttt{d}\xi_{n}\sigma(\xi')\texttt{d}x'\nonumber\\
\end{eqnarray}
Similarly to \textbf{Case a(II)} in \cite{Wa3}, we have
\begin{equation}
-\frac{1}{4}\int_{|\xi'|=1}\int_{-\infty}^{+\infty}\text{trace}[\partial_{x_{n}}\pi_{\xi_{n}}^{+}
\sigma_{-1}((D_{T}^{*})^{-1})\partial_{\xi_{n}}^{2}\sigma_{-1}(D_{T}^{-1})](x_{0})\texttt{d}\xi_{n}\sigma(\xi')\texttt{d}x'
=-\frac{3}{16}\pi h'(0)\Omega_{3}\texttt{d}x'.
\end{equation}
On the other hand, similarly to (2.2.18) in \cite{Wa3}, we have
\begin{eqnarray}
&&\partial_{x_{n}}\pi_{\xi_{n}}^{+} \omega^{g}\sigma_{-1}((D_{T}^{*})^{-1})(x_{0})|_{|\xi'|=1} \nonumber\\
&=&c(\tilde{e}_{1})c(\tilde{e}_{2})c(\tilde{e}_{3})c(\texttt{d}x_{n})\Big[\frac{\partial_{x_{n}}[c(\xi')](x_{0})}{2(\xi_{n}-i)}
          +\frac{(2i-\xi_{n})h'(0)c(\xi')}{4(\xi_{n}-i)^{2}}-\frac{h'(0)c(\texttt{d}x_{n})}{4(\xi_{n}-i)^{2}}\Big]
\end{eqnarray}
and
\begin{eqnarray}
\partial^2_{\xi_n}\sigma_{-1}(D_{T}^{-1})&=&\sqrt{-1}\Big(-\frac{6\xi_nc(dx_n)+2c(\xi')}
{|\xi|^4}+\frac{8\xi_n^2c(\xi)}{|\xi|^6}\Big) \nonumber\\
&=&\frac{6i\xi_n^{2}-2i}{(1+\xi_n^{2})^3}c(\xi')+\frac{2i\xi_n^{3}-6i\xi_n}{(1+\xi_n^{2})^3}c(dx_n).
\end{eqnarray}
By the relation of the Clifford action and $\texttt{tr}AB=\texttt{tr}BA$, considering for $i<n$ ,
$\int_{|\xi'|=1}\{\xi_{i_1}\xi_{i_2}\cdots\xi_{i_{2d+1}}\}\sigma(\xi')=0$, then
\begin{equation}
\texttt{tr}[c(\tilde{e}_{1})c(\tilde{e}_{2})c(\tilde{e}_{3})c(\xi')]=0, \ \
\texttt{tr}[c(\tilde{e}_{1})c(\tilde{e}_{2})c(\tilde{e}_{3})c(\texttt{d}x_{n})\partial_{x_{n}}[c(\xi')]c(\xi')]=0.
\end{equation}
Hence in this case,
\begin{equation}
\text{trace}[\partial_{\xi'}^{\alpha}\pi_{\xi_{n}}^{+} \omega^{g}\sigma_{-1}((D_{T}^{*})^{-1})\partial_{x'}^{\alpha}
\partial_{\xi_{n}}\sigma_{-1}(D_{T}^{-1})](x_{0})
=0.
\end{equation}
Therefore
\begin{equation}
{\rm Case~\bar{a})~II)}=-\frac{3}{16}\pi h'(0)\Omega_{3}\texttt{d}x',
\end{equation}
where $\Omega_{3}$ is the canonical volume of $S^{3}.$

\textbf{Case $\bar{a}$) III)}: \ $r=-1, \ \ell=-1, \ j=|\alpha|=0, \ k=1$

From (2.13), we have
 \begin{eqnarray}
{\rm Case~\bar{a})~III)}&=&-\frac{1}{2}\int_{|\xi'|=1}\int_{-\infty}^{+\infty}\text{trace}[\partial_{\xi_{n}}\pi_{\xi_{n}}^{+}
       P^{+}\sigma_{-1}((D_{T}^{*})^{-1})\partial_{\xi_{n}}\partial_{x_{n}}\sigma_{-1}(D_{T}^{-1})](x_{0})\texttt{d}\xi_{n}\sigma(\xi')
       \texttt{d}x'  \nonumber\\
&=&-\frac{1}{4}\int_{|\xi'|=1}\int_{-\infty}^{+\infty}\text{trace}[\partial_{\xi_{n}}\pi_{\xi_{n}}^{+}
       \sigma_{-1}((D_{T}^{*})^{-1})\partial_{\xi_{n}}\partial_{x_{n}}\sigma_{-1}(D_{T}^{-1})](x_{0})\texttt{d}\xi_{n}\sigma(\xi')\texttt{d}x'
        \nonumber\\
&&-\frac{1}{4}\int_{|\xi'|=1}\int_{-\infty}^{+\infty}\text{trace}[\partial_{\xi_{n}}\pi_{\xi_{n}}^{+}
      \omega^{g} \sigma_{-1}((D_{T}^{*})^{-1})\partial_{\xi_{n}}\partial_{x_{n}}\sigma_{-1}(D_{T}^{-1})](x_{0})\texttt{d}\xi_{n}\sigma(\xi')
      \texttt{d}x'.\nonumber\\
\end{eqnarray}
 Similarly to \textbf{Case a(III)} in \cite{Wa3}, we have
\begin{equation}
-\frac{1}{4}\int_{|\xi'|=1}\int_{-\infty}^{+\infty}\text{trace}[\partial_{\xi_{n}}\pi_{\xi_{n}}^{+}
       \sigma_{-1}((D_{T}^{*})^{-1})\partial_{\xi_{n}}\partial_{x_{n}}\sigma_{-1}(D_{T}^{-1})](x_{0})\texttt{d}\xi_{n}\sigma(\xi')\texttt{d}x'
       =\frac{3}{16}\pi h'(0)\Omega_{3}\texttt{d}x'.
\end{equation}
 Similarly to (2.2.27) in \cite{Wa3}, we have
\begin{equation}
\partial_{\xi_{n}}\pi_{\xi_{n}}^{+} \omega^{g}\sigma_{-1}((D^{*})^{-1})(x_{0})|_{|\xi'|=1}
=-c(\tilde{e}_{1})c(\tilde{e}_{2})c(\tilde{e}_{3})c(\texttt{d}x_{n})\Big[\frac{c(\xi')+ic(\texttt{d}x_{n})}{2(\xi_{n}-i)^{2}}\Big],
\end{equation}
and
\begin{equation}
\partial_{\xi_{n}}\partial_{x_{n}}\sigma_{-1}(D^{-1})(x_{0})|_{|\xi'|=1}
=-\sqrt{-1}h'(0)\Big[\frac{c(\texttt{d}x_{n})}{|\xi|^{4}}-4\xi_{n}\frac{c(\xi')+\xi_{n}c(\texttt{d}x_{n})}{|\xi|^{6}}\Big]
-\frac{2\sqrt{-1}\xi_{n}\partial_{x_{n}}c(\xi')(x_{0})}{|\xi|^{4}}.
\end{equation}
Similarly to (6.11), we have
\begin{equation}
\text{trace}[\partial_{\xi_{n}}\pi_{\xi_{n}}^{+} \omega^{g}\sigma_{-1}(D_{T}^{-1})\partial_{\xi_{n}}\partial_{x_{n}}\sigma_{-1}(D_{T}^{-1})](x_{0})
=0.
\end{equation}
Then
\begin{equation}
{\rm Case~\bar{a})~III)}=\frac{3}{16}\pi h'(0)\Omega_{3}\texttt{d}x',
\end{equation}
where $\Omega_{3}$ is the canonical volume of $S^{3}.$

Thus the sum of ${\rm Case~\bar{a})~II)}$ and ${\rm Case~\bar{a})~III)}$ is zero.

\textbf{Case $\bar{b}$)}: \ $r=-2, \ \ell=-1, \ k=j=|\alpha|=0$

By (2.13), we get
\begin{eqnarray}
{\rm Case~\bar{b})}&=&-i\int_{|\xi'|=1}\int_{-\infty}^{+\infty}\text{trace}[\pi_{\xi_{n}}^{+}P^{+}\sigma_{-2}((D_{T}^{*})^{-1})
                 \partial_{\xi_{n}}\sigma_{-1}(D_{T}^{-1})](x_{0})\texttt{d}\xi_{n}\sigma(\xi')\texttt{d}x' \nonumber\\
&=&-\frac{i}{2}\int_{|\xi'|=1}\int_{-\infty}^{+\infty}\text{trace}[\pi_{\xi_{n}}^{+}\sigma_{-2}((D_{T}^{*})^{-1})
                 \partial_{\xi_{n}}\sigma_{-1}(D_{T}^{-1})](x_{0})\texttt{d}\xi_{n}\sigma(\xi')\texttt{d}x'\nonumber\\
&&-\frac{i}{2}\int_{|\xi'|=1}\int_{-\infty}^{+\infty}\text{trace}[\pi_{\xi_{n}}^{+}\omega^{g}\sigma_{-2}((D_{T}^{*})^{-1})
                 \partial_{\xi_{n}}\sigma_{-1}(D_{T}^{-1})](x_{0})\texttt{d}\xi_{n}\sigma(\xi')\texttt{d}x'.
\end{eqnarray}
 Similarly to (3.51), we have
\begin{equation}
-\frac{i}{2}\int_{|\xi'|=1}\int_{-\infty}^{+\infty}\text{trace}[\pi_{\xi_{n}}^{+}\sigma_{-2}((D_{T}^{*})^{-1})
                 \partial_{\xi_{n}}\sigma_{-1}(D_{T}^{-1})](x_{0})\texttt{d}\xi_{n}\sigma(\xi')\texttt{d}x'
=\Big[\frac{9}{16}h'(0)-\frac{1}{4}\sum_{i}A_{iin}\Big]\pi \Omega_{3}\texttt{d}x'.
\end{equation}
By Lemma 3.4 and Lemma 3.5, we have
\begin{eqnarray}
\sigma_{-2}((D_{T}^{*})^{-1})(x_{0})&=&\frac{c(\xi)\sigma_{0}(D_{T}^{*})(x_{0})c(\xi)}{|\xi|^{4}}+\frac{c(\xi)}{|\xi|^{6}}\sum_{j}c(dx_{j})
                    \Big[\partial_{x_{j}}(c(\xi))|\xi|^{2}-c(\xi)\partial_{x_{j}}(|\xi|^{2})\Big](x_{0})\nonumber\\
                   &=&\frac{c(\xi)\sigma_{0}(D_{T}^{*})(x_{0})c(\xi)}{|\xi|^{4}}+
                  + \frac{c(\xi)}{|\xi|^{6}}c(dx_{n})\Big[\partial x_{n}(c(\xi'))(x_{0})-c(\xi)h'(0)|\xi'|^{2}_{g^{\partial M}}\Big].
\end{eqnarray}
Hence
\begin{equation}
\pi_{\xi_{n}}^{+} \omega^{g}\sigma_{-2}((D^{*})^{-1})(x_{0}):=\tilde{A}_{1}+\tilde{A}_{2},
\end{equation}
where
\begin{eqnarray}
\tilde{A}_{1}&=&c(\tilde{e}_{1})c(\tilde{e}_{2})c(\tilde{e}_{3})c(\texttt{d}x_{n})\Big\{
-\frac{h'(0)}{2}\left[\frac{c(dx_n)}{4i(\xi_n-i)}+\frac{c(dx_n)-ic(\xi')}{8(\xi_n-i)^2}
+\frac{3\xi_n-7i}{8(\xi_n-i)^3}[ic(\xi')-c(dx_n)]\right]  \nonumber\\
&&+\frac{-1}{4(\xi_n-i)^2}\Big[(2+i\xi_n)c(\xi')\alpha_0c(\xi')+i\xi_nc(dx_n)\alpha_0c(dx_n)+(2+i\xi_n)c(\xi')c(dx_n)\partial_{x_n}c(\xi')\nonumber\\
&&~~~~+ic(dx_n)\alpha_0c(\xi')+ic(\xi')\alpha_0c(dx_n)-i\partial_{x_n}c(\xi')\Big]\Big\};\\
\tilde{A}_{2}&=&c(\tilde{e}_{1})c(\tilde{e}_{2})c(\tilde{e}_{3})c(\texttt{d}x_{n})\Big\{\frac{-1}{4(\xi_n-i)^2}\Big[(2+i\xi_n)c(\xi')\beta_0c(\xi')
                +i\xi_nc(dx_n)\beta_0c(dx_n)+ic(dx_n)\beta_0c(\xi')\nonumber\\
 &&+ic(\xi')\beta_0c(dx_n)\Big]\Big\}
\end{eqnarray}
and
\begin{eqnarray}
\alpha_0&=&-\frac{3}{4}h'(0)c(dx_n),\\
\beta_0&=&\frac{1}{4}\sum_{i\neq s\neq t}A_{ist}
  c(\widetilde{e_{i}})c(\widetilde{e_{s}})c(\widetilde{e_{t}}) -\frac{1}{4}\sum_{i, s, t}\Big[-A_{iit}c(\widetilde{e_{t}})
  +A_{isi}c(\widetilde{e_{s}}) -A_{iss}c(\widetilde{e_{i}})+2A_{iii}c(\widetilde{e_{i}})\Big].
\end{eqnarray}
On the other hand, a simple computation shows
\begin{equation}
\partial_{\xi_{n}}\sigma_{-1}(D_{T}^{-1})=\frac{-2 i \xi_{n}c(\xi')}{(1+\xi_{n}^{2})^{2}}+\frac{i(1- \xi_{n}^{2})c(\texttt{d}x_{n})}
{(1+\xi_{n}^{2})^{2}}.
\end{equation}
From (6.23) and (6.27), we obtain
\begin{equation}
-\frac{i}{2}\int_{|\xi'|=1}\int_{-\infty}^{+\infty}\text{trace}[\tilde{A}_{1}\times
                 \partial_{\xi_{n}}\sigma_{-1}(D_{T}^{-1})](x_{0})\texttt{d}\xi_{n}\sigma(\xi')\texttt{d}x'
                 =0.
\end{equation}
Combining (6.24) and (6.27), we have
\begin{equation}
\text{trace}[\tilde{A}_{2}\times\partial_{\xi_{n}}\sigma_{-1}(D_{T}^{-1})](x_{0})
   =\frac{-i \tilde{c}_0}{2(\xi_{n}-i)^{2}(\xi_{n}+i)^{2}},
\end{equation}
where
\begin{equation}
\tilde{c}_0= 2(A_{123}-A_{213}+A_{312}).
\end{equation}
From (6.19) and (6.29), we get
\begin{eqnarray}
&&-\frac{i}{2}\int_{|\xi'|=1}\int_{-\infty}^{+\infty}\text{trace}[\tilde{A}_{2}\times
                 \partial_{\xi_{n}}\sigma_{-1}(D_{T}^{-1})](x_{0})\texttt{d}\xi_{n}\sigma(\xi')\texttt{d}x'\nonumber\\
&=&-\frac{i}{2}\Omega_{3}\int_{\Gamma^{+}}\frac{-i \tilde{c}_0}{2(\xi_{n}-i)^{2}(\xi_{n}+i)^{2}}\texttt{d}\xi_{n}\texttt{d}x' \nonumber\\
&=&-\frac{i}{2} \frac{2\pi i }{1!}\tilde{c}_0\Omega_{3}\Big[\frac{-i }{2(\xi_{n}+i)^{2}}\Big]^{(1)}|_{\xi_{n}=i}\texttt{d}x'  \nonumber\\
&=&-\frac{1}{8}\pi\tilde{c}_0\Omega_{3}\texttt{d}x'.
\end{eqnarray}
Combining (6.20) and (6.31) , we have
\begin{equation}
{\rm Case~\bar{b})}=\Big[\frac{9}{16}h'(0)-\frac{1}{4}(A_{123}-A_{213}+A_{312})
-\frac{1}{4}\sum_{i}A_{iin}\Big]\pi \Omega_{3}\texttt{d}x'.
\end{equation}

\textbf{Case $\bar{c}$}: \ $r=-1, \ \ell=-2, \ k=j=|\alpha|=0$

From (2.13), we have
 \begin{eqnarray}
{\rm Case~\bar{c})}&=&-i\int_{|\xi'|=1}\int_{-\infty}^{+\infty}\text{trace}[\pi_{\xi_{n}}^{+}P^{+}\sigma_{-1}((D_{T}^{*})^{-1})
 \partial_{\xi_{n}}\sigma_{-2}(D_{T}^{-1})](x_{0})\texttt{d}\xi_{n}\sigma(\xi')\texttt{d}x'  \nonumber\\
&=&-\frac{i}{2}\int_{|\xi'|=1}\int_{-\infty}^{+\infty}\text{trace}[\pi_{\xi_{n}}^{+}\sigma_{-1}((D_{T}^{*})^{-1})
 \partial_{\xi_{n}}\sigma_{-2}(D_{T}^{-1})](x_{0})\texttt{d}\xi_{n}\sigma(\xi')\texttt{d}x'  \nonumber\\
&&-\frac{i}{2}\int_{|\xi'|=1}\int_{-\infty}^{+\infty}\text{trace}[\pi_{\xi_{n}}^{+}\omega^{g}\sigma_{-1}((D_{T}^{*})^{-1})
 \partial_{\xi_{n}}\sigma_{-2}(D_{T}^{-1})](x_{0})\texttt{d}\xi_{n}\sigma(\xi')\texttt{d}x'.
\end{eqnarray}
 Similarly to (3.63), we have
\begin{equation}
-\frac{i}{2}\int_{|\xi'|=1}\int_{-\infty}^{+\infty}\text{trace}[\pi_{\xi_{n}}^{+}\sigma_{-1}((D_{T}^{*})^{-1})
 \partial_{\xi_{n}}\sigma_{-2}(D_{T}^{-1})](x_{0})\texttt{d}\xi_{n}\sigma(\xi')\texttt{d}x'
=\Big[-\frac{9}{16}h'(0)-\frac{1}{4}\sum_{i}A_{iin}\Big]\pi \Omega_{3}\texttt{d}x'.
\end{equation}
By Lemma 3.5, we have
\begin{equation}
\pi_{\xi_{n}}^{+} \omega^{g}\sigma_{-1}((D_{T}^{*})^{-1})=\Big(c(\tilde{e}_{1})c(\tilde{e}_{2})c(\tilde{e}_{3})c(\texttt{d}x_{n}) \Big)
       \frac{c(\xi')+ic(\texttt{d}x_{n})}{2(\xi_{n}-i)}.
\end{equation}
By Lemma 3.4 and Lemma 3.5, we have
\begin{eqnarray}
\sigma_{-2}((D_{T})^{-1})(x_{0})&=&\frac{c(\xi)\sigma_{0}(D_{T})(x_{0})c(\xi)}{|\xi|^{4}}+\frac{c(\xi)}{|\xi|^{6}}\sum_{j}c(dx_{j})
                    \Big[\partial_{x_{j}}(c(\xi))|\xi|^{2}-c(\xi)\partial_{x_{j}}(|\xi|^{2})\Big](x_{0})\nonumber\\
                   &=&\frac{c(\xi)\sigma_{0}(D_{T})(x_{0})c(\xi)}{|\xi|^{4}}+
                  + \frac{c(\xi)}{|\xi|^{6}}c(dx_{n})\Big[\partial x_{n}(c(\xi'))(x_{0})-c(\xi)h_{x_{n}}'(0)|\xi'|^{2}_{g^{\partial M}}\Big].
\end{eqnarray}
Then
\begin{equation}
 \partial_{\xi_{n}}\sigma_{-2}(D_{T}^{-1})(x_{0}):=B_{1}+B_{2},
\end{equation}
where
\begin{eqnarray}
B_{1}&=&\frac{1}{(1+\xi_n^2)^3}\Big[(2\xi_n-2\xi_n^3)c(dx_n)\alpha_0c(dx_n)+(1-3\xi_n^2)c(dx_n)\alpha_0c(\xi')\nonumber\\
    &&+ (1-3\xi_n^2)c(\xi')\alpha_0c(dx_n)-4\xi_nc(\xi')\alpha_0 c(\xi')+(3\xi_n^2-1)\partial_{x_n}c(\xi')
    -4\xi_nc(\xi')c(dx_n)\partial_{x_n}c(\xi')\nonumber\\
    &&+2h'(0)c(\xi')+2h'(0)\xi_nc(dx_n)\Big]
    +6\xi_nh'(0)\frac{c(\xi)c(dx_n)c(\xi)}{(1+\xi^2_n)^4},\\
B_{2}&=&\frac{1}{(1+\xi_n^2)^3}\Big[(2\xi_n-2\xi_n^3)c(dx_n)\beta_1c(dx_n)+(1-3\xi_n^2)c(dx_n)\beta_1c(\xi')\nonumber\\
    &&+(1-3\xi_n^2)c(\xi')\beta_1c(dx_n)-4\xi_nc(\xi')\beta_1c(\xi')\Big]
\end{eqnarray}
and
\begin{eqnarray}
\beta_1&=&\frac{1}{4}\sum_{i\neq s\neq t}A_{ist}
  c(\widetilde{e_{i}})c(\widetilde{e_{s}})c(\widetilde{e_{t}})+\frac{1}{4}\sum_{i, s, t}[-A_{iit}c(\widetilde{e_{t}})
  +A_{isi}c(\widetilde{e_{s}}) -A_{iss}c(\widetilde{e_{i}})+2A_{iii}c(\widetilde{e_{i}})].
\end{eqnarray}
Similar to (6.28), we have
\begin{equation}
-\frac{i}{2}\int_{|\xi'|=1}\int_{-\infty}^{+\infty}\text{trace}[\pi_{\xi_{n}}^{+}\sigma_{-1} \omega^{g}((D_{T}^{*})^{-1})\times B_{1}]
               (x_{0})\texttt{d}\xi_{n}\sigma(\xi')\texttt{d}x'
                 =0.
\end{equation}
From (6.35) and (6.39), we obtain
\begin{equation}
\text{trace}[\pi_{\xi_{n}}^{+}\sigma_{-1} \omega^{g}((D_{T}^{*})^{-1})\times B_{2}](x_{0})
=\frac{i \tilde{c}_0}{(\xi_{n}+i)^{3}(\xi_{n}-i)},
\end{equation}
where
\begin{equation}
\tilde{c}_0= 2(A_{123}-A_{213}+A_{312}).
\end{equation}
By (6.33) and (6.42), we get
\begin{eqnarray}
&&-\frac{i}{2}\int_{|\xi'|=1}\int_{-\infty}^{+\infty}\text{trace}[\pi_{\xi_{n}}^{+}\sigma_{-1} \omega^{g}((D_{T}^{*})^{-1})\times B_{2}]
               (x_{0})\texttt{d}\xi_{n}\sigma(\xi')\texttt{d}x'\nonumber\\
&=&-\frac{i}{2}\Omega_{3}\int_{\Gamma^{+}}\frac{i \tilde{c}_0}{(\xi_{n}+i)^{3}(\xi_{n}-i)}\texttt{d}\xi_{n}\texttt{d}x' \nonumber\\
&=&-\frac{i}{2} 2\pi i \tilde{c}_0 \Omega_{3}[\frac{i }{(\xi_{n}+i)^{3}}]^{(0)}|_{\xi_{n}=i}\texttt{d}x'  \nonumber\\
&=&-\frac{1}{8}\pi\tilde{c}_0\Omega_{3}\texttt{d}x',
\end{eqnarray}
From (6.34) and (6.44)  we obtain
\begin{equation}
{\rm Case~\bar{c})}=\Big[-\frac{9}{16}h'(0)-\frac{1}{4}(A_{123}-A_{213}+A_{312})
-\frac{1}{4}\sum_{i}A_{iin}\Big]\pi \Omega_{3}\texttt{d}x'.
\end{equation}

Now $\Phi$  is the sum of the ${\rm Case~\bar{a})~\bar{b})~ and ~\bar{c})}$, so
\begin{equation}
\Phi=-\frac{1}{2}\Big[(A_{123}-A_{213}+A_{312})
+\sum_{i}A_{iin}\Big]\pi \Omega_{3}\texttt{d}x'.
\end{equation}
Hence we conclude that
\begin{thm}
 Let M be a 4-dimensional compact manifold  with the boundary $\partial M$  and let $\tilde{R}$ be the scalar curvature
of the modified connection $\tilde{\nabla}$. Then we get the volumes  associated to $P^{+}D^{*}D$,
 \begin{equation}
\widetilde{Vol}_{4}^{(1, 1)}=-\frac{1}{96\pi^{2}}\int_{M}(\tilde{R}\omega^{g}+\tilde{C}_{H})
  -\frac{1}{2}\int_{\partial_{M}}\Big[(A_{123}-A_{213}+A_{312})
+\sum_{i}A_{iin}\Big]\pi \Omega_{3}\texttt{d}x'.
\end{equation}
where $\tilde{R}=R+18div(V)-54|V|^{2}-9\parallel T\parallel^{2}$ be the scalar curvature of the modified connection $\tilde{\nabla}$
and $\tilde{C}_{H}=18(dT-\langle T, *V\rangle\omega^{g})$.
\end{thm}

\section{ The Kastler-Kalau-Walze type theorem for 4-dimensional complex manifolds  associated with complex nonminimal operators  }

In this section, we compute the lower dimension volume for lower dimension compact connected manifolds with boundary and get a
Kastler-Kalau-Walze type Formula in this case.

Let $M$ be a compact Riemannian manifold of dimension $m$ without boundary. If $M$ is
equipped with integrable complex structure, one can split tangential indices into holomorphic and
antiholomorphic ones and define space of differential forms  $C^{\infty}(\Lambda^{p,q})$. The exterior differential $d$
can be also split into a sum $d=\partial+\bar{\partial}$ of anticommuting nilpotent operators:
$\partial^{2}=\bar{\partial}^{2}=\partial\bar{\partial}+\bar{\partial}\partial=0$.
If M is a K$\ddot{a}$hler manifold, the corresponding ``Laplacians" can be reduced to the De-Rham Hodge Laplacian:
 \begin{equation}
\partial^{*}\partial+\partial\partial^{*}=\overline{\partial\partial^{*}}+\overline{\partial^{*}\partial}=\frac{1}{2}\Delta=
\frac{1}{2}(d\delta+\delta d).
\end{equation}
Using these first order differential operators one can construct a  nonminimal second order
differential operator:
 \begin{equation}
\mathfrak{D}=g_{1}\partial\partial^{*} +g_{2}\partial^{*}\partial+g_{3}\overline{\partial\partial^{*}}+g_{4}\overline{\partial^{*}\partial}
+g_{5}\partial\bar{\partial^{*}}+g_{5}^{*}\bar{\partial}\partial^{*}
\end{equation}
with real constants $g_{1},\cdots, g_{4}$ and a complex constant $g_{5}$. For some values of the constants this
operator reduces to that considered previously in this paper. One can find some motivations for studying nonminimal operators.
Such operators appear naturally in quantum gauge theories after imposing gauge conditions.

For complex manifold of dimension $4$ with boundary. By Proposition 2.1 in \cite{JMB}, let $D_{T}=\sqrt{2}(\bar{\partial}+\bar{\partial^{*}})$
we have the identity  of three forms $T=\sqrt{-1}(\partial-\bar{\partial})\omega$, where
$\omega$ be the K$\ddot{a}$hler forms. In this case, we have $A_{iin}=0$. Then
\begin{cor}
 Let M be a 4-dimensional compact complex manifold  with the boundary $\partial M$ and $\tilde{\nabla}$ be an orthogonal
connection with torsion. Then we get the volumes  associated to $D_{T}^{2}$,
 \begin{equation}
Vol_{4}^{(1, 1)}=-\frac{1}{48\pi^{2}}\int_{M}\tilde{R}(x)dx,
\end{equation}
where $\tilde{R}=R+18div(V)-54|V|^{2}-9\parallel T\parallel^{2}$ and $\int_{M}div(V) dVol_{M}=-\int_{\partial_{M}}g(n,V)dVol_{\partial_{M}}$.
\end{cor}
Next for the operator $(\bar{\partial}+\bar{\partial^{*}})^{2}$, we consider the heat kernel for nonminimal operators acting on
the space $C^{\infty}(\Lambda^{k})$ of $k$ forms.
Let us discuss the first order operators $D_{1}=\partial$ and $D_{2}=\bar{\partial}$ which satisfy the properties of Lemma 1.
And these operators  will be used to build up the following general nonminimal second order operator
 \begin{equation}
\mathfrak{D}=a^{2} \overline{\partial\partial^{*}}+b^{2}\overline{\partial^{*}\partial}.
\end{equation}
The nonminimal operator with real constants $a^{2},b^{2}$ is the most general hermitian operator on $C^{\infty}(\Lambda^{k})$
which can be constructed using $\partial,\bar{\partial},\partial^{*}$ and $\bar{\partial^{*}}$. This operator has the form (7.2).
Denote by $\sigma_{l}(\mathfrak{D})$ the $l$-order symbol of an operator $\mathfrak{D}$.  We compute the symbol expansion of
 $\mathfrak{D}=a^{2} \overline{\partial\partial^{*}}+b^{2}\overline{\partial^{*}\partial}$.
Recall \cite{Y} , we have

\begin{lem}
The following equalities hold
\begin{eqnarray}
&&\sigma_L(\overline{\partial})(x,\xi)=\frac{\sqrt{-1}}{2}\sum_{1\leq j\leq n}(\xi_j+\sqrt{-1}\xi_{j+n})\varepsilon(e^j-\sqrt{-1}e^{j+n});\\
&&\sigma_L(\overline{\partial}^*)(x,\xi)=-\frac{\sqrt{-1}}{2}\sum_{1\leq j\leq n}(\xi_j-\sqrt{-1}\xi_{j+n})\iota(e^j-\sqrt{-1}e^{j+n});\\
&&\sigma_L(\tilde{\triangle})=\frac{1}{2}|\xi|^2{\rm Id},
\end{eqnarray}
where $\tilde{\triangle}= \overline{\partial\partial^{*}}+\overline{\partial^{*}\partial}$.
\end{lem}

For $\xi=\sum_{1\leq i \leq 2n}\xi_ie^i$ , let
$\widehat{\xi}=\sum_{1\leq j\leq n}(\xi_j+\sqrt{-1}\xi_{j+n})(e^j-\sqrt{-1}e^{j+n}).$
For a $(0,1)$-form $\omega_1$, we have
\begin{eqnarray}
\iota^{g_{i}}((\xi_j+\sqrt{-1}\xi_{j+n})(e^j-\sqrt{-1}e^{j+n}))\omega_1
&=&\langle\omega_1,(\xi_j+\sqrt{-1}\xi_{j+n})(e^j-\sqrt{-1}e^{j+n})\rangle \nonumber\\
 &=&(\xi_j-\sqrt{-1}\xi_{j+n})\iota^{g_{i}}(e^j-\sqrt{-1}e^{j+n})\omega_1.
\end{eqnarray}
By Lemma 7.2 and (7.8), we have
\begin{equation}
\sigma_L(\overline{\partial})(\xi)=\frac{\sqrt{-1}}{2}\varepsilon(\widehat{\xi});~
\sigma_L(\overline{\partial}^*)(\xi)=-\frac{\sqrt{-1}}{2}\iota(\widehat{\xi}).
\end{equation}
Then we obtain the  following Lemma.
\begin{lem}
Let $\mathfrak{D}=a^{2} \overline{\partial\partial^{*}}+b^{2}\overline{\partial^{*}\partial}$ on $C^{\infty}(\Lambda^{k})$, then
\begin{equation}
\sigma_{-2}(\mathfrak{D}^{-1})=\frac{2 b^{2}|\xi|^{2}+(a^{2}-b^{2})\iota(\hat{\xi})\varepsilon(\hat{\xi})}{a^{2}b^{2}|\xi|^{4}}.
\end{equation}
\end{lem}
\begin{proof}
From Lemma 7.2, we have
$\sigma_{2}(\overline{\partial\partial^{*}})=\frac{1}{4}\varepsilon(\hat{\xi})\iota(\hat{\xi}), \
 \sigma_{2}(\overline{\partial\partial^{*}}+\overline{\partial^{*}\partial})=\frac{1}{2}|\xi|^{2}.$
Combining these results, we obtain
\begin{equation}
 \sigma_{2}(a^{2} \overline{\partial\partial^{*}}+b^{2}\overline{\partial^{*}\partial})
 =(a^{2}-b^{2})\sigma_{2}(\overline{\partial\partial^{*}})+b^{2}\sigma_{2}(\overline{\partial\partial^{*}}+\overline{\partial^{*}\partial})
=\frac{(a^{2}-b^{2})}{4}\varepsilon(\hat{\xi})\iota(\hat{\xi}) +\frac{b^{2}}{2}|\xi|^{2}.
\end{equation}
An application of $2|\xi|^{2}=\varepsilon(\hat{\xi})\iota(\hat{\xi})+\iota(\hat{\xi})\varepsilon(\hat{\xi})$ shows
\begin{equation}
\Big[\frac{(a^{2}-b^{2})}{4}\varepsilon(\hat{\xi})\iota(\hat{\xi}) +\frac{b^{2}}{2}|\xi|^{2}\Big]
\Big[\frac{(a^{2}-b^{2})}{4}\iota(\hat{\xi})\varepsilon(\hat{\xi}) +\frac{b^{2}}{2}|\xi|^{2}\Big]=\frac{a^{2}b^{2}}{4}|\xi|^{4}.
\end{equation}
Then
\begin{equation}
\sigma_{-2}(\mathfrak{D}^{-1})=\frac{2 b^{2}|\xi|^{2}+(a^{2}-b^{2})\iota(\hat{\xi})\varepsilon(\hat{\xi})}{a^{2}b^{2}|\xi|^{4}}.
\end{equation}
\end{proof}

Without loss of generality,  we may assume the differential
operator $\mathfrak{D}=a^{2} \overline{\partial\partial^{*}}+b^{2}\overline{\partial^{*}\partial}$ acting  on $C^{\infty}(\Lambda^{k})$.
An application of Theorem 1 in \cite{SADV}  yields the following:

\begin{thm}\cite{SADV}
For m-dimensional $(m>2)$ compact K$\ddot{a}$hler manifolds without boundary and the associated nonminimal operator
$\mathfrak{D}=a^{2} \overline{\partial\partial^{*}}+b^{2}\overline{\partial^{*}\partial}$ on $C^{\infty}(\Lambda^{k})$, then
\begin{eqnarray}
a_{2}[\mathfrak{D}|_{C^{\infty}(\Lambda^{k})}]
&=&\Big(\frac{a^{2}}{2}\Big)^{2-\frac{m}{2} }\sum_{p=0}^{k-2}(-1)^{k-p}(k-p-1)a_{2}(\Delta_{p})  \nonumber\\
&&-\Big[\Big(\frac{a^{2}}{2}\Big)^{2-\frac{m}{2} }+\Big(\frac{b^{2}}{2}\Big)^{2-\frac{m}{2} }\Big]
 \sum_{p=0}^{k-1}(-1)^{k-p}(k-p)a_{2}(\Delta_{p})  \nonumber\\
&&+\Big(\frac{b^{2}}{2}\Big)^{2-\frac{m}{2} }\sum_{p=0}^{k}(-1)^{k-p}(k-p+1)a_{2}(\Delta_{p}) .
\end{eqnarray}
\end{thm}
From the Theorem 4.8.18 in \cite{PBG} we obtain
\begin{thm}\cite{PBG}
Let the $\Delta^{m}_{p}=d\delta+\delta d $ denote the Laplacian acting on the space of
smooth p-forms on an m-dimensional manifold. We let $R_{ijkl}$ denote the
curvature tensor with the sign convention that $R_{1212}=-1$ on the sphere
of radius 1 in $R^{3}$. Then:
\begin{eqnarray}
&&a_{0}(\Delta^{m}_{p})=(4\pi)^{-\frac{m}{2}}C_{m}^{p}; \\
&&a_{2}(\Delta^{m}_{p})=\frac{(4\pi)^{-\frac{m}{2}}}{6}\Big(C_{m-2}^{p-2} +C_{m-2}^{p} -4C_{m-2}^{p-1} \Big)(-R_{ijij}).
\end{eqnarray}
\end{thm}
Then for $m\geq 4$ the coefficients are
\begin{eqnarray}
&&a_{2}(\Delta_{0}^{m})=\frac{(4\pi)^{-\frac{m}{2}}}{6}(-R_{ijij}); \\
&&a_{2}(\Delta_{1}^{m})=\frac{(4\pi)^{-\frac{m}{2}}}{6}(6-m)(-R_{ijij});\\
&&a_{2}(\Delta_{2}^{m})=\frac{(4\pi)^{-\frac{m}{2}}}{12}(-m^{2}+13m-24)R_{ijij};\\
&&a_{2}(\Delta_{3}^{m})=\frac{(4\pi)^{-\frac{m}{2}}}{36}(108-92m+21m^{2}-m^{3})R_{ijij};\\
&&a_{2}(\Delta_{4}^{m})=\frac{(4\pi)^{-\frac{m}{2}}}{144}(-576+630m-227m^{2}+30m^{3}-m^{4})R_{ijij}.
\end{eqnarray}
Let $k=4$ (the calculation of other  case similar), combining these results, we obtain the main theorem of this section.
\begin{thm}
For 4-dimension  compact K$\ddot{a}$hler manifolds without boundary and the associated nonminimal operator
$\mathfrak{D}=a^{2} \overline{\partial\partial^{*}}+b^{2}\overline{\partial^{*}\partial}$ on $C^{\infty}(\Lambda^{4})$, then
\begin{eqnarray}
Wres(\mathfrak{D}^{-1})&=&\frac{(m-2)(2\pi)^{4}}{\Gamma (\frac{m-k}{2})}\Big\{\Big[\Big(\frac{a^{2}}{2}\Big)^{2-\frac{m}{2} }-
\Big(\frac{b^{2}}{2}\Big)^{2-\frac{m}{2} }\Big]\frac{(4\pi)^{-\frac{m}{2}}}{36}(222-167m+24m^{2}-m^{3})R_{ijij}\nonumber\\
&&+\Big(\frac{b^{2}}{2}\Big)^{2-\frac{m}{2}}\frac{(4\pi)^{-\frac{m}{2}}}{144}(-576+630m-227m^{2}+30m^{3}-m^{4})R_{ijij}\Big\}|_{m=4}\nonumber\\
&=&-\frac{\pi^{2}}{3}R_{ijij},
\end{eqnarray}
  where  $R_{ijkl}$ denote the curvature tensor with the sign convention that $R_{1212}=-1$ on the sphere of radius 1 in $R^{3}$.
\end{thm}

\section*{ Acknowledgements}
This work was supported by Fok Ying Tong Education Foundation under Grant No. 121003 and NSFC. 11271062. And Jian Wang's Email
address: wangj068@gmail.com. The author also thank the referee for his (or her) careful reading and helpful comments.

\end{document}